\documentclass[11pt]{article}
\usepackage{suffix}
\usepackage{mathtools}
\usepackage{array}
\usepackage{faktor}
\usepackage{wrapfig, amsmath}
\usepackage{amsfonts}
\usepackage{amssymb}
\usepackage{bm}
\usepackage{graphicx}
\usepackage{epsfig,cancel,amsthm, amssymb, todonotes, amsmath}
\usepackage{color,tikz}
\usetikzlibrary{decorations.markings,arrows, calc}
\usepackage{color}
\usepackage{graphicx,framed,verbatim,caption,amsbsy}
\usepackage[colorlinks=true, pdfstartview=FitV, linkcolor=darkblue, citecolor=darkblue, urlcolor=darkblue]{hyperref}
\usepackage{enumitem}
\usepackage[rflt]{floatflt}
\usepackage[T1]{fontenc}
\usepackage[utf8]{inputenc}
\usepackage{bigints}

\definecolor{shadecolor}{rgb}{0.9, 0.9, 0.86}
\definecolor{darkgreen}{rgb}{0.2, 0.5,  0}
\definecolor{darkblue}{rgb}{0.1,0.1,0.45}

\def\&{\vspace{-5pt}&}

\def\Re{\mathrm {Re}\,}
\def\Im{\mathrm {Im}\,}

\usepackage{tikz}
\usepackage{pgf}
\usepackage{pgflibraryarrows}
\usepackage{pgflibrarysnakes}
\usetikzlibrary{decorations.text}
\usepgfmodule{shapes}
\usetikzlibrary{decorations.pathmorphing}
\usetikzlibrary{decorations.markings}
\usetikzlibrary{patterns}
\usetikzlibrary{automata}
\usetikzlibrary{positioning}
\usepackage{authblk}

\tikzset{->-/.style={decoration={
 markings,
 mark=at position #1 with {\arrow{>}}},postaction={decorate}}}

\def \eqref#1{(\ref{#1})}
\def \& {&\hspace{-10pt}}

\renewcommand{\d}{\mathrm d}
       
\newtheorem{theorem}{Theorem}[section]
\newtheorem{example}[theorem]{Example}
\newtheorem{exercise}[theorem]{Exercise}

\newtheorem{lemma}[theorem]{Lemma}
\newtheorem{remark}[theorem]{Remark}

\newtheorem{proposition}[theorem]{Proposition} 
\newtheorem{corollary}[theorem]{Corollary} 
 
\newtheorem{definition}[theorem]{Definition}
\newtheorem{assumption}[theorem]{Assumption}

\def\d{{\rm d}}

\def\e{{\rm e}}

\def\bt{\begin{theorem}}
\def\et{\end{theorem}}
\def\bc{\begin{corollary}}
\def\ec{\end{corollary}}
\def\bx{\begin{example}}
\def\ex{\end{example}}
\def\bxr{\begin{exercise}\small}
\def\exr{\end{exercise}}
\def\bl{\begin{lemma}}
\def\el{\end{lemma}}
\def\bd{\begin{definition}}
\def\ed{\end{definition}}
\def\bp{\begin{proposition}}
\def\ep{\end{proposition}}

\def\br{\begin{remark}}
\def\er{\end{remark}}

\def\be{\begin{equation}}
\def\ee{\end{equation}}
\def \Tr {\mathrm{Tr}\,}
\def\&{\hspace{-15pt}&}
\def\bea{\begin{eqnarray}}
\def\eea{\end{eqnarray}}
\def\beas{\begin{eqnarray*}}
\def\eeas{\end{eqnarray*}}

\def\R{{\mathbb R}}

\def\1{{\bf 1}}

\def\i{\mathrm{i}}
\def\valpha{{\vec \alpha}}
\def\va{{\vec a}}
\DeclarePairedDelimiterX\MeijerM[3]{\lparen}{\rparen}%
{\begin{smallmatrix}#1 \\ #2\end{smallmatrix}\delimsize\vert\,#3}

\newcommand\MeijerG[8][]{%
  G^{\,#2,#3}_{#4,#5}\MeijerM[#1]{#6}{#7}{#8}}

\WithSuffix\newcommand\MeijerG*[7]{%
  G^{\,#1,#2}_{#3,#4}\MeijerM*{#5}{#6}{#7}}

\makeatletter
\@addtoreset{equation}{section}
\makeatother

\begin{document}

\date{}                     

\title{Biorthogonal measures, polymer partition functions, and random matrices}
\author[1]{Mattia Cafasso}
\author[2]{Tom Claeys}
\renewcommand\Affilfont{\small}
\affil[1]{\textit{Univ Angers, SFR MATHSTIC, F-49000 Angers, France;} \texttt{mattia.cafasso@univ-angers.fr}}
\affil[2]{\textit{Institut de Recherche en Math\'ematique et Physique, UCLouvain, Chemin du Cyclotron 2, 1348 Louvain-la-Neuve, Belgium;} \texttt{tom.claeys@uclouvain.be}}
                                  
\maketitle
\begin{abstract}
We develop the study of a particular class of biorthogonal measures, encompassing at the same time several random matrix models and partition functions of polymers. This general framework allows us to characterize the partition functions of the Log Gamma polymer and the mixed polymer in terms of explicit biorthogonal measures, as it was previously done for the homogeneous O'Connell-Yor polymer by Imamura and Sasamoto. In addition, we show that the biorthogonal measures associated to these three polymer models (Log Gamma, O'Connell-Yor, and mixed polymer) converge to random matrix eigenvalue distributions in small temperature limits. We also clarify the connection between different Fredholm determinant representations and explain how our results might be useful for asymptotic analysis and large deviation estimates.
\end{abstract}


\section{Introduction}\label{sec:Intro}

 We study a class of biorthogonal measures on $\mathbb R^N$, characterized by kernels admitting double contour integral representations. Specific measures of this form are known to be eigenvalue distributions of random matrices, for instance in the Gaussian, Laguerre, and Jacobi Unitary Ensembles (GUE, LUE, and JUE) \cite{Forrester, Mehta}, as well as in generalizations of these like the GUE with external source \cite{BrezinHikami, BrezinHikami2}, the LUE with external source \cite{BBP, ElKaroui}, sums of LUE and GUE matrices \cite{ClaeysKuijlaarsWang}, product random matrix ensembles \cite{AkemannIpsenKieburg, AkemannKieburgWei, Forrester2, ForresterLiu, KieburgKuijlaarsStivigny, KuijlaarsStivigny, KuijlaarsZhang}, and Muttalib-Borodin ensembles \cite{BorBiOE}. Another measure of this form characterizes the partition function of the O'Connell-Yor polymer partition function \cite{ImamuraSasamoto1}. The first purpose of this paper is to study general properties of this class of biorthogonal measures. This will allow us to interpret all of the above models as instances of a unified framework, and to identify other models which fit in our scheme. More precisely, we will prove that the Log Gamma polymer \cite{Seppalainen} and the mixed polymer \cite{BCFV} partition functions can also be characterized in terms of explicit biorthogonal measures. 
These characterizations are direct analogues of a result by Imamura and Sasamoto \cite{ImamuraSasamoto1} for the { homogeneous} O'Connell-Yor polymer. Moreover, they can be seen as finite temperature generalizations of the connection between last passage percolation with exponential weights and the LUE \cite{Johansson2, BorodinPeche, DiekerWarren}, and of the connection between Brownian queues and the GUE \cite{Baryshnikov, GravnerTracyWidom} (see also \cite{BaikRains} for a different last passage percolation model). From another perspective, they are  finite $N$ variants of the connection between the Kardar-Parisi-Zhang (KPZ) equation and the Airy point process \cite{ACQ}.
In addition, the biorthogonal measures enable us to clarify the connection between various Fredholm determinant expressions which appeared in the literature before, and to derive a novel representation of these Fredholm determinants, which is possibly more convenient in view of asymptotic analysis and large deviation estimates.

\medskip

A probability measure on $\mathbb R^N$ of the form
\be\label{def:biO}\frac{1}{Z_N}\det\left(f_m(x_k)\right)_{k,m=1}^N
\det\left(g_m({x_k})\right)_{m,k=1}^N
\prod_{k=1}^N\d  x_k,
\ee
for certain sets of functions $f_1,\ldots, f_N$ and $g_1,\ldots, g_N$,
is called a biorthogonal ensemble \cite{BorBiOE}. Observe that there is a lot of freedom in the choices of $f_1,\ldots, f_N$ and $g_1,\ldots, g_N$ for a given biorthogonal ensemble, since one can take linear combinations of columns and rows of the matrices $\left(f_m(x_k)\right)_{k,m=1}^N$ and $\left(g_m(x_k)\right)_{m,k=1}^N$ without affecting the distribution.
Biorthogonal ensembles are special cases of determinantal point processes: for $n=1,\ldots, N$, the $n$-point correlation function 
\be\label{def:rhok}\rho_{n,N}(x_1,\ldots, x_n):=\frac{N!}{(N-n)! Z_N}\int_{\mathbb R^{N-n}}
\det\left(f_m(x_k)\right)_{k,m=1}^N
\det\left(g_m({x_k})\right)_{m,k=1}^N
\prod_{k=n+1}^N\d  x_k
\ee can be expressed as a determinant:
\be\label{def:corrfunction}\rho_{n,N}(x_1,\ldots, x_n)=\det\left(L_N(x_m,x_k)\right)_{m,k=1}^n.\ee
The correlation kernel $L_N$ can be taken of the form
\be\label{def:corrkernel}L_N(x,x')=\sum_{k=1}^NF_k(x)G_k(x'),\ee
where $F_1,\ldots, F_N$ have the same linear span as $f_1,\ldots, f_N$, $G_1,\ldots, G_N$ have the same linear span as $g_1,\ldots, g_N$, and they satisfy the biorthogonality relations
\be\label{def:bio}
\int_{\mathbb R}F_k(x)G_m(x)\d x=\delta_{km},\qquad k,m=1,\ldots, N.
\ee
Average multiplicative statistics in biorthogonal ensembles can be expressed as Fredholm series: for any bounded measurable $\sigma:\mathbb R\to\mathbb C$, we have
\be
\label{eq:avgmult}
\mathbb E \left[ \prod_{k=1}^N(1-\sigma(x_k)) \right]=\sum_{k=0}^N \frac{(-1)^k}{k!}\int_{\mathbb R^k}\det\left(L_N(x_j,x_m)\right)_{j,m=1}^k\prod_{j=1}^k\sigma(x_j)\d x_j.
\ee
The right hand side is the (truncated) Fredholm series associated to the kernel $\sigma(x)L_N(x,x')$, which we will denote henceforth as $\det(1-\sigma L_N)_{L^2(\mathbb R)}$. The subscript $L^2(\mathbb R)$ indicates that we are integrating over $\mathbb R$, however at this point, we do not interpret $\sigma L_N$ yet as a trace-class operator acting on $L^2(\mathbb R)$, and
we use $\det(1-\sigma L_N)_{L^2(\mathbb R)}$ simply as a short-hand notation for the right hand side of \eqref{eq:avgmult}.

\medskip

{\em Orthogonal polynomial ensembles} \cite{Konig} form an important subclass of biorthogonal ensembles, corresponding to
$f_m(x)=x^{m-1}$ and $g_m(x)=x^{m-1}{w(x)}$ for some weight function $w$. Then we can take $F_k=P_{k-1}$ to be the degree $k-1$ normalized orthogonal polynomial on the real line with respect to the weight $w$, and $G_k=P_{k-1}w$. Orthogonal polynomial ensembles arise for instance as eigenvalue distributions of unitarily invariant Hermitian random matrices, see e.g.\ \cite{Deift}.
{\em Polynomial ensembles} \cite{KuijlaarsPol} are more general and correspond to $f_m(x)=x^{m-1}$ with arbitrary $g_m$. Examples of them arise as distributions of singular values of products of random matrices and in Muttalib-Borodin ensembles \cite{BorBiOE}.
{\em P\'olya ensembles} \cite{ForsterKieburgKosters, KieburgZhang} are polynomial ensembles with a specific derivative structure, corresponding to $f_m(x)=x^{m-1}$ and $g_m(x)=\displaystyle\frac{\d^{m-1}}{\d x^{m-1}}g(x)$ for some function $g$.

\medskip

In this paper, we  consider a class of complex-valued measures on $\mathbb R^N$, $N\in \mathbb N$, of the following form
\be\label{def:BiOE-F}
\d\mu_N(\vec{x}) \equiv \d\mu_N(x_1,\ldots, x_N):=\frac{1}{N!}
\det\left(L_N(x_m,x_k)\right)_{m,k=1}^N
\prod_{k=1}^N\d  x_k,
\ee
where the kernel $L_N$ can be represented as a double contour integral of the form
\begin{equation}\label{def:L}
		L_N(x,x') = \frac{1}{(2 \pi \i)^2} \int_{\Sigma_N} \d u \int_{\ell_N} \d v \frac{W_N(v)}{W_N(u)} \frac{{\rm e}^{-vx + ux'}}{v - u}.
	\end{equation}
Here, $W_N$, $\ell_N$, and ${\Sigma_N}$ need to satisfy the following assumptions.
\begin{assumption}\label{assumptions}There exists a vertical strip $\mathcal S_N:=\{z\in\mathbb C:\alpha_N<\Re z<\beta_N\}$ in the complex plane such that the following holds:
\begin{enumerate}\item $W_N:\mathcal S_N\to\mathbb C$ is analytic, not identically zero and such that $W_N(z)=O(|z|^{-\epsilon})$ as $z\to\infty$ in $\mathcal S_N$, for some $\epsilon>0$;
\item $\ell_N$ is a vertical line in $\mathcal S_N$, oriented upwards, i.e.\ $\ell_N=c_N+\i\mathbb R$ with $\alpha_N<c_N<\beta_N$;
\item ${\Sigma_N}$ is a closed positively oriented curve in $\mathcal S_N$ without self-intersections, lying at the left of $\ell_N$, which has $N$ (not necessarily distinct) zeros of $W_N$ in its interior, and none on its image; we denote these zeros as $a_1,\ldots, a_N$.
\end{enumerate}
\end{assumption}
Note that we allow $W_N$ to have other zeros in $\mathcal S_N$ than $a_1,\ldots, a_N$, and that the zeros $a_1,\ldots, a_N$ do not need to be distinct. The expression \eqref{def:L} is independent of the choice of ${\Sigma_N}$ and $\ell_N$, as long as the above assumptions are satisfied.
Under these assumptions, we will show that the complex-valued measures \eqref{def:BiOE-F} are biorthogonal: they can be written in the form \eqref{def:biO}. We will use the term {\em biorthogonal measure} instead of {\em biorthogonal ensemble} to indicate that the measures are in general not probability measures. {All the examples we study, nevertheless, are real-valued, even if not necessarily positive.}
We also define, in analogy with \eqref{def:rhok}, marginals for such measures,
\begin{equation}\label{def:complexcorrelators}
	\rho_{n,N}(x_1,\ldots,x_n) := \frac{1}{(N - n)!}\int_{\mathbb R^{N - n}} \det\Big(L_N(x_m,x_k)\Big)_{m,k=1}^N \prod_{k = n +1}^N \mathrm d x_{k}.
\end{equation} 
Note that we use the same normalization as for the correlation functions given in \eqref{def:rhok}, however we avoid the term correlation functions because the measure $\mu_N$ is not necessarily positive.

\medskip

Our first result contains the essence of the biorthogonal structure of the measures \eqref{def:BiOE-F}.

\begin{theorem}\label{thm:biorth}
Let $W_N, {\Sigma_N}, \ell_N$ satisfy Assumption \ref{assumptions}, and let $\mu_N$ be defined by \eqref{def:BiOE-F}.
\begin{enumerate}
\item The kernel $L_N$ given by \eqref{def:L} 
satisfies the properties
\be\label{eq:reproducing0}\int_{\mathbb R}L_N(x,t)L_N(t,x')\d t=L_N(x,x'),\qquad \int_{\mathbb R}L_N(x,x)\d x=N,\ee
and, with $\rho_{n,N}$ as in \eqref{def:complexcorrelators}, we have
\be\rho_{n,N}(x_1,\ldots, x_n)=\det\Big(L(x_m,x_k)\Big)_{m,k=1}^n,\qquad n=1,\ldots, N.\ee
\item If $a_1, \ldots, a_N$ are distinct, we have
\be\label{def:L1}
L_N(x,x')=\sum_{m=1}^N \e^{a_m x'}\psi_m(\e^{x}),\quad \text{where} \quad
\psi_m(y):=\frac{1}{2\pi\i W_N'(a_m)}\int_{\ell_N}\frac{W_N(v)y^{-v}}{v-a_m}\d v,
\ee
and we have the biorthogonality relations
\be\label{eq:biorth}\int_{\mathbb R}\e^{a_mx}\psi_k(\e^x)\d x=\delta_{mk},\qquad k,m=1,\ldots, N.\ee
{The measure \eqref{def:BiOE-F} is then a biorthogonal measure, given explicitly by}
\be\label{BiOE-distinct}
\d\mu_N(\vec x)=\frac{1}{N!}\det\left(\e^{a_m x_k}\right)_{k,m=1}^N
\det\left(\psi_m(\e^{x_k})\right)_{m,k=1}^N
\prod_{k=1}^N\d  x_k.\ee
\item In the confluent case $a_1=\cdots= a_N=a$, {the measure \eqref{def:BiOE-F} is also biorthogonal and given explicitly by}
\be\label{BiOE-confluent}
\d\mu_N(\vec x)=\frac{1}{Z_N}\Delta(\vec x)
\det\left(\phi_m(\e^{x_k})\right)_{m,k=1}^N
\prod_{k=1}^N\d  x_k,\quad \Delta(\vec x)=\prod_{1\leq j < k\leq N}(x_k-x_j),\ee
where
\be
 Z_N:=\int_{\mathbb R^N}\Delta(\vec x)
\det\left(\phi_m(\e^{x_k})\right)_{m,k=1}^N
\prod_{k=1}^N\d  x_k,
\ee
and 
\be\label{def:phiF}
\phi_m(y):=\frac{y^a}{2\pi\i}\int_{\ell_N}\frac{W_N(v)y^{-v}}{(v-a)^{N-m+1}}\d v.
\ee
\end{enumerate}
\end{theorem}
\begin{remark}\label{rem:smalltmm}
As we will explain in detail in Sections \ref{section:LogGamma}, \ref{section:OY} and \ref{section:mixed}, positive biorthogonal measures of the form \eqref{def:BiOE-F}--\eqref{def:L} occur in a variety of random matrix models. Three choices of $W$ are particularly relevant in this paper: they are given by
\begin{equation}
\begin{array}{rcl}
	W_N^{{\rm LUE+}}(z;{\vec b}, \nu) &=& \displaystyle \frac{\prod_{j = 1}^N (z - {b_j})}{(z - 1)^{N + \nu}},\\
	\\
	W_N^{{\rm GUE+}}(z;\vec a, \tau) &=& \displaystyle {\rm e}^{\tau z^2/2}\prod_{j = 1}^N (z - a_j)\\
	\\
	W_N^{{\rm GLUE+}}(z;{\vec b}, \tau) &=& \displaystyle {\frac{z^N{\rm e}^{\tau z^2/2}}{\prod_{j = 1}^N (z-1+b_j)}}.
\end{array}
\end{equation}
They are associated, respectively, to the biorthogonal ensembles given by the eigenvalues of a LUE (Laguerre Unitary Ensemble) matrix with external source, a GUE (Gaussian Unitary Ensemble) with external source, and a sum of a LUE matrix with external source and a GUE matrix. {As such, the choices of $W$ given above define positive-valued measures}. We anticipate that these three models describe the small temperature limits of three different polymer models: Log Gamma, O'Connell-Yor and O'Connell-Yor with boundary sources. More details about these small temperature limits are given at the end of the introduction.
\end{remark}

\begin{remark}
Observe that the integral expressions for $\psi_m$ and $\phi_m$
are independent of the choice of the vertical line $\ell_N$ in the strip $\mathcal S_N$. 
Recalling the definition of the direct and inverse Mellin transforms,
\be\label{def:Mellin}
\mathcal M[g](v):=\int_0^{\infty}y^{v-1}g(y)\d y,
\qquad\mathcal M^{-1}[G](y):=\frac{1}{2\pi\i}\int_{\ell_N}{G(v)y^{-v}}\d v,\ee
we recognize $\psi_m$ as the inverse Mellin transforms of the function $\frac{W_N(v)}{W_N'(a_m)(v-a_m)}$ and $\phi_m$ as the inverse Mellin transform of $W_N(v)/(v-a)^{N-m+1}$ multiplied by $y^a$.
In exponential variables, 
$\psi_1,\ldots, \psi_N$ and $\phi_1,\ldots, \phi_m$ are Fourier integrals:
\be\label{def:psiF1}
\psi_m(\e^x)=\frac{1}{2\pi\i W_N'(a_m)}\int_{\ell_N}\frac{W_N(v)}{v-a_m}\e^{-v x}\d v=\frac{\e^{-c x}}{2\pi W_N'(a_m)}\int_{\mathbb R}\frac{W_N(c+\i t)}{c+\i t-a_m}\e^{-\i t x}\d t\ee
for $\alpha_N<c<\beta_N$ and $x\in\mathbb R$, and
similarly for $\phi_1,\ldots, \phi_m$.
\end{remark}
\begin{remark}
The confluent form of the biorthogonal measure has the structure of a complex-valued (not necessarily positive) polynomial ensemble, and more specifically of a P\'olya ensemble, since $-\frac{\partial}{\partial x} \phi_m(\e^x) = \phi_{m+1}(\e^x)$. It will become clear from the proof that explicit {biorthogonal} expressions for the measure can be given not only in the completely confluent case where all of the values $a_1,\ldots, a_N$ coincide, but also in cases where some of them occur with multiplicity bigger than one.
\end{remark}
\begin{remark}\label{remark:exp}
In exponential variables $s_j=\e^{x_j}$, $j=1,\ldots, N$, the biorthogonal measure \eqref{BiOE-distinct} becomes
\be\label{BiOE-distinct-exp}
\d\widehat\mu_N(\vec s):=\frac{1}{N!}\det\left(s_k^{a_m-1}\right)_{k,m=1}^N
\det\left(\psi_m(s_k)\right)_{m,k=1}^N
\prod_{k=1}^N\d  s_k,\qquad \vec s\in(0,\infty)^N, 
\ee
with associated kernel 
\begin{equation}\label{def:Lhat}
		\widehat L_N(s,s') = \frac{1}{(2 \pi \i)^2} \int_{\Sigma_N} \d u \int_{\ell_N} \d v \frac{W_N(v)}{W_N(u)} \frac{{s}^{-v-1}s'^{u}}{v - u}.
	\end{equation}
In the special case where $a_m=m$, $m=1,\ldots, N$, the first determinant becomes a Vandermonde determinant, and then we have the polynomial ensemble
\be\label{BiOE-distinct-exp-VdM}
\d\widehat\mu_N(\vec s):=\frac{1}{N!}\Delta(\vec s)\ 
\det\left(\psi_m(s_k)\right)_{m,k=1}^N
\prod_{k=1}^N\d  s_k,\qquad \vec s\in(0,\infty)^N.
\ee
Note that the mechanism to obtain a polynomial ensemble is different here than in the confluent case, as it takes place for equi-spaced $a_1,\ldots, a_N$, and in exponential variables. For specific choices of $W_N$ which are ratios of products of Gamma functions, this measure is positive and describes the distribution of singular values for products of random matrices and Muttalib-Borodin ensembles, as we will explain at the end of Section \ref{section:mixed}.
\end{remark}

The above biorthogonal structure of the measures \eqref{def:BiOE-F} gives us access to determinantal representations for averages of multiplicative statistics of the form $\prod_{k=1}^N(1-\sigma(x_k))$. 
The simplest is an expression in terms of an $N\times N$ determinant. In addition, we have several Fredholm determinant expressions.
Specific Fredholm determinants of these types appear in random matrix theory and in the literature of polymer models and Whittaker processes, see for instance \cite{BCF,BCFV, BorodinGorin}.
In our next result, we will gather all these different determinantal representations for average multiplicative statistics. This will allow us in particular to identify models of a different nature, random matrix ensembles and polymer models, as instances of our class of biorthogonal measures.

We will denote henceforth $\mu_N[\sigma]$ for the average multiplicative statistic associated to a function $\sigma$:
\be\mu_N[\sigma]:=\int_{\mathbb R^N}\prod_{k=1}^N(1-\sigma(x_k))\d\mu_N(x_1,\ldots, x_N).\label{def:avgmultstat}\ee

\medskip

Given a curve $\gamma$ in the complex plane, we recall the definition of the Fredholm determinant of an operator via the Fredholm series associated to its kernel:
\be\label{def:Fredholm}
\det(1-K)_{L^2(\gamma)}=\sum_{k=0}^\infty \frac{(-1)^k}{k!}\int_{\gamma^k}\det\left(K(z_j,z_m)\right)_{j,m=1}^k\prod_{j=1}^k\d z_j.
\ee
The relevant curves for us will be $\gamma=\mathbb R$, $\gamma=(0,+\infty)$, and $\gamma={\Sigma_N}$.
As already mentioned, for now, we write
$\det(1-K)_{L^2(\gamma)}$ whenever \eqref{def:Fredholm} is convergent, and we do not require the associated operator to be trace-class. Later on, we will see that suitable conjugations $h(x)K(x,x')h(x')^{-1}$ of the relevant kernels $K$ give rise to rank $N$ and thus trace-class operators on $L^2(\gamma)$ (in that case the series truncates after the term $k=N$), which will allow us to use classical properties for Fredholm determinants.

\medskip

A first  identity expresses $\mu_N[\sigma]$ in terms of the Fredholm determinant with kernel $L_N$, and holds in general under Assumptions \ref{assumptions}, for bounded and measurable $\sigma$. A second identity requires stronger decay of $W_N$ and sufficiently fast decay of $\sigma$ at $-\infty$, and involves 
the Fredholm determinant with kernel
\be
\label{def:H}
H_N^\sigma(y,y') := \int_{\mathbb R} \sigma(x)\Psi_1(\e^{y+x})\Psi_2(\e^{y'+x})  \d x,\qquad y,y'>0,
\ee
where
\be\label{def:psi12}\Psi_1(s) := \frac{1}{2\pi\i}\int_{{\Sigma_N}} \frac{s^{u}}{W_N(u)}\d u,\qquad \Psi_2(s) := \frac{1}{2 \pi \i}\int_{\ell_N}  s^{-v} W_N(v)\d v=\mathcal M^{-1}[W_N](s).\ee
A third identity holds only for the specific choice $\sigma_t(x)=\frac{1}{1+\e^{-x-t}}$ for $t\in\mathbb R$, requires 
that all zeros $a_1,\ldots, a_N$ are contained in a vertical strip of width $<1$, i.e.\ $a_{\max}-a_{\min}<1$, with
\be\label{def:aminmax}a_{\min}:=\min\{\Re a_1,\ldots, \Re a_N\},\qquad a_{\max}:=\max\{\Re a_1,\ldots, \Re a_N\},\ee
and involves the Fredholm determinant with kernel 
\be\label{def:kernel}
K_{N,t}(u,u') := \frac{1}{(2 \pi \i)^2} \int_{\ell_N} \d v \frac{\pi \e^{t(v-u)}}{\sin \pi(u - v)} \frac{W_N(v)}{W_N(u)}\frac{1}{v - u'},\qquad u,u'\in{\Sigma_N}.
\ee

\begin{theorem}\label{thm:Fredholm}
Let $W_N, {\Sigma_N}, \ell_N$ satisfy Assumption \ref{assumptions}, let $\sigma:\mathbb R\to \mathbb C$ be bounded and measurable, and let $\mu_N$ be defined by \eqref{def:BiOE-F}-\eqref{def:L}.
\begin{enumerate}
\item 
In the case where $a_1,\ldots, a_N$ are distinct, we have, with $\psi_k$ as in \eqref{def:L1},
\[\mu_N[\sigma]=\det\left(\int_{\mathbb R}\left(1-\sigma(x)\right)\e^{a_m x}\psi_k(\e^x)\d x\right)_{m,k=1}^N.\]
If $a_1=\cdots = a_N=a$, then with $\phi_k$ as in \eqref{def:phiF},
\[\mu_N[\sigma]=\frac{\det\left(\displaystyle\int_{\mathbb R}\left(1-\sigma(x)\right)x^{m-1}\phi_k(\e^x)\d x\right)_{m,k=1}^N}{\det\left(\displaystyle\int_{\mathbb R}x^{m-1}\phi_k(\e^x)\d x\right)_{m,k=1}^N}.\]
\item We have the Fredholm series representation
\be\label{eq:Fredholm1}
\mu_N[\sigma]=\det(1-\sigma L_N)_{L^2(\mathbb R)},
\ee
where $L_N$ is given by \eqref{def:L}.
\item If there exists $\epsilon>0$ such that
\[W_N(v)=O(|v|^{-1-\epsilon}),\quad v\to\infty, \, v\in\mathcal S_N,\qquad \sigma(x)=O\left(\e^{(a_{\max}-a_{\min}+\epsilon)x}\right),\quad x\to -\infty,\]
we have the alternative Fredholm series representation
\be\label{eq:Fredholm2}
\mu_N[\sigma]=\det(1-H_N^\sigma)_{L^2(0,\infty)},
\ee
with $H_N^\sigma$ as in \eqref{def:H}.
\item 
If there exists $\epsilon>0$ such that
\[W_N(v)=O(|v|^{-1-\epsilon}),\quad v\to\infty, \, v\in\mathcal S_N,\]
$a_{\max}-a_{\min}<1$
and $\sigma_t(x)=\frac{1}{1+\e^{-x-t}}$ for $t\in\mathbb R$, we have yet another Fredholm determinant representation:
\be\label{eq:Fredholm3}
\mu_N[\sigma_t]=\det(1+K_{N,t})_{L^2({\Sigma_N})},
\ee
with $K_{N,t}$ given by \eqref{def:kernel}.
\end{enumerate}
\end{theorem}
\begin{remark}
The kernel $H_N^\sigma$ has the structure of a finite-temperature type deformation of the kernel $L_N$. Indeed, $L_N$ is a special case of $H_N^\sigma$, corresponding to the case when $\sigma$ is the indicator function of the interval $[0,+\infty)$.
Such deformations appear frequently, for example, in the study of free fermions at finite temperature \cite{LeDoussalMajumdarSchehr}, the KPZ equation \cite{ACQ}, and in the study of Schur processes with periodic boundary conditions \cite{BeteaBouttier, BorodinCylindrical}.
\end{remark}
\begin{remark}\label{remark:large}
The result of part 4 may at first sight seem too specific to be interesting, but as we will point out later, it is precisely this identity that will allow us to
make the connection with the Log Gamma polymer, since
Fredholm determinants with kernels of the form \eqref{def:kernel} are intimately connected with ($\alpha$-)Whittaker processes and Sklyanin measures, see e.g.\ \cite{BC, BCR}. In particular, a Fredholm determinant of this type is known to be equal to the Laplace transform of the Log Gamma polymer partition function \cite{BCR}. It should be noted that Fredholm determinant expressions like \eqref{eq:Fredholm3}, but with unbounded contours ${\Sigma_N}$ containing an infinite number of zeros of $W_N$, are also frequently encountered in the literature of polymers and Whittaker processes, see e.g.\ \cite{BCF, BCFV} and the discussion in \cite[Section 2]{BCD}. However, such expressions with kernels defined on {\em large contours} do not fit directly in our framework.
\end{remark}

\medskip

Next, we introduce a deformation of the biorthogonal measure \eqref{def:BiOE-F}. Given a bounded measurable function $\sigma:\mathbb R\to\mathbb C$, we define the measures
\be\label{def:BiOE-F-deformed}
\d\nu_{N,t}(\vec x):=\frac{1}{Z_{N,t}}
\det\Big(L_N(x_m,x_k)\Big)_{m,k=1}^N
\prod_{k=1}^N\big(1-\sigma_t(x_k)\big)\d  x_k,\quad t\in\mathbb R,\ \sigma_t(x) := \sigma(x+t),
\ee
where \be\label{def:ZNp}Z_{N,t} := \int_{\mathbb R^N}\det\Big(L_N(x_m,x_k)\big)_{m,k=1}^N
\prod_{k=1}^N\big(1-\sigma_t(x_k)\big)\d  x_k,\ee 
so that $\displaystyle\int_{\mathbb R^N}\d\nu_{N,t}(x_1,\ldots, x_N)=1$, and where we assume that $\sigma, N,t$ are such that $Z_{N,t}\neq 0$. 

For $\sigma(x)=1_{(0,\infty)}(x)$, the deformed measure is simply the original biorthogonal measure, restricted to $(-\infty, -t)^N$ and renormalized or, in other words, the original biorthogonal ensemble conditioned on $\max\{x_1,\ldots, x_N\}\leq -t$, in case the measure is positive. For general $\sigma$, if the biorthogonal measure is positive, it can be seen as a pushed particle system \cite{KrajenbrinkLeDoussal} in which the particles $x_1,\ldots, x_N$ are pushed to regions where $\sigma_t(x)$ is small; it can be constructed via a procedure of marking and conditioning of the original biorthogonal ensemble \cite{ClaeysGlesner}. 

By definition, the deformed measure \eqref{def:BiOE-F-deformed}  is also a biorthogonal measure. However, the associated kernel in general does not admit an explicit double integral form like in \eqref{def:L}, which makes the deformed measure harder to analyze. Nevertheless, the deformed measure encodes $\mu_N[\sigma]$ in a simple way:
as we will see, logarithmic derivatives of average multiplicative statistics $\mu_N[\sigma]$ of the non-deformed measure can be expressed in terms of average linear statistics 
of the deformed measure. Average linear statistics of determinantal measures are much easier to analyze than average multiplicative statistics, since they can be expressed in terms of the one-point functions. Similarly to the one-point function $\rho_{1,N}(x)$ defined by \eqref{def:rhok} and given by \eqref{def:corrfunction}, we define the one-point function of the deformed measure as
\be\label{def:onepoint}\kappa_{N,t}(x_1)=\frac{N}{Z_{N,t}}(1-\sigma_t(x_1))\int_{\mathbb R^{N-1}}\det\left(L_N(x_m,x_k)\right)_{m,k=1}^N
\prod_{k=2}^N(1-\sigma_t(x_k))\d  x_k.\ee
It is then straightforward to verify that
\be\label{def:linstat}
\bigintsss_{\mathbb R^N}\left(\sum_{k=1}^N f(x_k)\right)\d\nu_{N,t}(x_1,\ldots, x_N)=\int_{\mathbb R}f(x)\kappa_{N,t}(x)\d x,
\ee
for any bounded and measurable $f:\mathbb R\to\mathbb C$.
Later on, it will be convenient to interpret
\eqref{def:BiOE-F-deformed} as an absolutely continuous measure on $\mathbb R^N$ with respect to the deformed measure $(1-\sigma_t(x))\d x$ instead of the Lebesgue measure. With respect to this deformed measure, the one-point function becomes
\be\label{def:onepoint2}
\widetilde\kappa_{N,t}(x_1)=\frac{N}{Z_{N,t}}\int_{\mathbb R^{N-1}}\det\big(L_N(x_m,x_k)\big)_{m,k=1}^N
\prod_{k=2}^N(1-\sigma_t(x_k))\d  x_k.
\ee
We have $\widetilde\kappa_{N,t}(x)=\kappa_{N,t}(x)/(1-\sigma_t(x))$ if $\sigma_t(x)\neq 0$, but $\widetilde\kappa_t(x)$ has the advantage that it is also defined when $\sigma_t(x)=1$.

\begin{theorem}\label{thm: deformed}
Let $W_N,{\Sigma_N},\ell_N$ satisfy Assumption \ref{assumptions}, let $\sigma:\mathbb R \to\mathbb C$ be bounded and $C^1$ with bounded derivative, and define $\sigma_t(x) := \sigma(x+t)$. Let $\d\mu_N$ and $\d\nu_{N,t}$ be defined by \eqref{def:BiOE-F} and \eqref{def:BiOE-F-deformed}, {and $\mu_N[\sigma]$ defined by \eqref{def:avgmultstat}}.
If $\mu_N[\sigma_t]\neq 0$, we have the identity
\be\label{eq:logder}
\frac{\d}{\d t}\log\mu_N[\sigma_t]=\int_{\mathbb R}\sigma'(x+t)\widetilde\kappa_{N,t}(x)\d x.
\ee
If $\sigma(x)=\displaystyle\frac{1}{1+\e^{-x}}$, the above identity simplifies to
\be\label{eq:logder2}
\frac{\d}{\d t}\log\mu_N[\sigma]=\int_{\mathbb R}\frac{1}{1+\e^{-x-t}}\kappa_{N,t}(x)\d x.
\ee
\end{theorem}
\begin{remark}
We believe that the last identity \eqref{eq:logder2} will be useful to study the asymptotic behavior of average multiplicative statistics $\mu_N[\sigma]$ as $N\to\infty$. Via \eqref{eq:logder2}, these asymptotics are encoded simply in the one-point function $\kappa_{N,t}(x)$ of the deformed biorthogonal measure. Even if the deformed biorthogonal measure is not particularly simple, there is a variety of techniques available to study one-point functions, like Coulomb gas techniques and characterizations via equilibrium problems. For specific choices of $W_N$, see Sections \ref{section:LogGamma}, \ref{section:OY} and \ref{section:mixed}, the above result implies that the one-point function $\kappa_{N,t}$ encodes tail asymptotics and large deviations of polymer partition functions.
The condition that $\sigma$ is $C^1$ is convenient for the proof of \eqref{eq:logder}, but it is not a necessary condition, and one can generalize the identity to piecewise $C^1$-functions. For instance, for $\sigma_t=1_{(-t,+\infty)}$, the identity \eqref{eq:logder} continues to hold if we replace $\sigma'$ by a Dirac $\delta$-function, such that
 \be\label{eq:logder3}
\frac{\d}{\d t}\log\mu_N[1_{(-t,\infty)}]=\widetilde\kappa_{N,t}(-t).
\ee
\end{remark}

\paragraph{Outline.} The rest of this paper is organized as follows. In Section \ref{section:biorthogonal}, we will study the biorthogonal structure of the measures $\mu_N$ and prove Theorem \ref{thm:biorth}. In Section \ref{section:det}, we will establish the determinantal expressions for the multiplicative statistics $\mu_N[\sigma]$, and prove Theorem \ref{thm:Fredholm}. In Section \ref{section:deformations}, we will  study the deformed biorthogonal measures $\nu_{N,t}$ and prove Theorem \ref{thm: deformed}. Finally, Sections \ref{section:LogGamma}, \ref{section:OY}, and \ref{section:mixed} are devoted to applications in polymer and random matrix models. There, we will study the biorthogonal measures \eqref{def:BiOE-F}--\eqref{def:L} associated to functions of the form
\begin{equation}
\begin{array}{rcl}
	\displaystyle W_{n,N}^{{\rm Log}\,\Gamma}(z;\vec\alpha,\vec a)&=& \displaystyle\frac{\prod_{j=1}^n\Gamma(\alpha_j-z)}{\prod_{k=1}^N\Gamma(z-a_k)},\\
	\\
	\displaystyle W_N^{{\rm OY}}(z;\vec a,\tau) &=& \displaystyle\frac{{\rm e}^{\frac{\tau z^2}2}}{\prod_{k = 1}^N \Gamma(z - a_k)},\\
	\\
	\displaystyle W_{N}^{{\rm Mixed}}(z;\vec\alpha,\vec a,\tau) &=&\displaystyle {\rm e}^{\tau z^2/2}\frac{\prod_{j=1}^N\Gamma(\alpha_j-z)}{\prod_{k=1}^N\Gamma(z-a_k)}
\end{array}
\end{equation}
and { prove that the associated measures are real-valued (but not necessarily positive) and characterize, respectively, the partition functions of the Log Gamma, the O'Connell-Yor, and the mixed polymer, see Corollaries \ref{finalcor:LogGamma}, \ref{finalcor:OY}, and \ref{finalcor:mixed}}. Moreover, we also prove that in a certain ``small temperature'' regime, these biorthogonal measures degenerate to the ones described in Remark \ref{rem:smalltmm} and associated to matrix models, see Propositions \ref{cor:LogGamma}, \ref{cor:OY} and \ref{cor:mixed}. For the Log Gamma and O'Connell-Yor polymers, this is not surprising, since the small temperature degeneration of the Log Gamma polymer is a last passage percolation model well-known to be characterized in terms of the LUE with external source \cite{Johansson2, BorodinPeche, DiekerWarren}, while the small temperature degeneration of the O'Connell-Yor polymer is a model of Brownian queues well-known to be characterized by the GUE with external source \cite{Baryshnikov,GravnerTracyWidom}. Nevertheless, it reveals how the zero temperature limit takes place on the level of the biorthogonal measures. For the mixed polymer, to the best of our knowledge, the degeneration to the matrix model given by the sum of LUE with external source and GUE is completely new.

\section{Biorthogonal structure of the measures}\label{section:biorthogonal}

We will suppose throughout this section that $W_N,{\Sigma_N},\ell_N$ satisfy Assumption \ref{assumptions}, and that $L_N$ is defined by \eqref{def:L}. The goal of this section is to prove Theorem \ref{thm:biorth}.

\medskip

To start, we focus on the situation where the zeros $a_1,\ldots, a_N$ of $W_N$ inside ${\Sigma_N}$ are all distinct. Then, we prove the following.
\begin{proposition}\label{prop:biorth}
If $a_1,\ldots, a_N$ are all distinct, we have
\be\label{eq:BiOL}
	L_N(x,x')=\sum_{m=1}^N\e^{a_mx'}\psi_m(\e^{x}),
\ee
with $\psi_m$ as in \eqref{def:psiF1}, and we have the orthogonality relation \eqref{eq:biorth}. Consequently, $L_N$ has the reproducing property
\be\label{eq:reproducing}\int_{\mathbb R}L_N(x,t)L_N(t,x')\d t=L_N(x,x').\ee
\end{proposition}
\begin{proof}
We can evaluate the $u$-integral in \eqref{def:L} using the residue theorem, and in this way, we immediately obtain \eqref{eq:BiOL}.

\medskip

Next, since
\[\psi_m(\e^x)=\frac{1}{2\pi\i W_N'(a_m)}\int_{\ell_N}\frac{W_N(v)\e^{-xv}}{v-a_m}\d v,\]
for any vertical line $\ell_N$ (avoiding $a_m$) in the strip $\mathcal S_N=\{z\in\mathbb C:\alpha_N<\Re z<\beta_N\}$, and since $W_N(z)=O(|z|^{-\epsilon})$ as $z\to\infty$,
it is easy to verify that, for any $\delta>0$,
\be\label{eq:aspsi}\psi_m(\e^x)=O\left(\e^{-(\alpha_N+\delta)x}\right)\quad\mbox{as $x\to -\infty$},\qquad \psi_m(\e^x)=O\left(\e^{-(\beta_N-\delta)x}\right)\quad\mbox{as $x\to +\infty$}.\ee
The integrals $\displaystyle\int_{-\infty}^{+\infty}\e^{a_k x}\psi_m(\e^{x})\d x$, $k=1,\ldots, N$ are thus well-defined, since $a_1,\ldots,a_N\in\mathcal S_N$.

\medskip

To obtain the biorthogonality relations, recalling the direct and inverse Mellin transform from \eqref{def:Mellin}, we observe that
$$
	\int_{-\infty}^{+\infty}\e^{a_k x}\psi_m(\e^{x})\d x = \int_{0}^\infty s^{a_k - 1}\psi_m(s) \d s = \frac{1}{W_N'(a_k)}\left(\mathcal M \circ\mathcal M^{-1} \right)\left[\frac{W_N(\cdot)}{\cdot - a_m}\right](a_k).
$$
Since $W_N$ has simple zeros at the points $a_1,\ldots, a_N$, the above expression is $0$ for $k\neq m$, and to $1$ for $k=m$.

\medskip

Finally, it is straightforward to deduce the reproducing property \eqref{eq:reproducing} from \eqref{eq:BiOL} and \eqref{eq:biorth}.
\end{proof}

As a consequence of \eqref{eq:BiOL}, using the multiplicativity of the determinant, we can rewrite
\eqref{def:BiOE-F} as
\begin{align}\d\mu_N(\vec x)&=\frac{1}{N!}\det\left(\sum_{j=1}^N\e^{a_jx_k}\psi_j(\e^{x_m})\right)_{m,k=1}^N\prod_{k=1}^N\d x_k\nonumber\\
&=\frac{1}{N!}\det\left(\e^{a_jx_k}\right)_{k,j=1}^N\det\left(\psi_j(\e^{x_m})\right)_{j,m=1}^N\prod_{k=1}^N\d x_k,\label{eq:biorthproduct}\end{align}
which is \eqref{BiOE-distinct}, such that part 2 of Theorem \ref{thm:biorth} is proved.

\begin{proposition}\label{prop:projection}
For arbitrary $a_1,\ldots, a_N$, $L_N$ satisfies the reproducing property \eqref{eq:reproducing} and we have $\displaystyle\int_{\mathbb R}L_N(x,x)\d x=N$.
\end{proposition}
\begin{proof}
The reproducing property is already proved in the case where all $a_1,\ldots, a_N$ are distinct.
	If not all $a_1,\ldots, a_N$ are distinct, write
	\[W_N(z)=W_0(z)\prod_{m=1}^N(z-a_m).\] For the sake of clarity, let us temporarily make explicit the dependence of $L_N$ on the parameters $a_1,\ldots,a_N$ in our notations, and write
	\[L_{N}(x,x';\vec{a}) \equiv L_N(x,x') = \frac{1}{(2 \pi \i)^2} \int_{\Sigma_N} \d u \int_{\ell_N} \d v \frac{W_0(v)\prod_{m=1}^N(v-a_m)}{W_0(u)\prod_{m=1}^N(u-a_m)} \frac{{\rm e}^{-vx + ux'}}{v - u}.\]
	Given $W_0$ and ${\Sigma_N}$, the kernel $L_N$ is continuous in $\vec a$, as long as $a_1,\ldots, a_N$ remain inside ${\Sigma_N}$. We already know that $L_{N}$ has the reproducing property if $a_1,\ldots, a_N$ are distinct:  
	\be\label{eq:reproducing2}\int_{\mathbb R}L_{N}(x,t;\vec a)L_{N}(t,x';\vec a)\d t=L_{N}(x,x';\vec a).\ee
Both sides of this identity are continuous in $\vec a$, thus the identity continues to hold when $a_1,\ldots, a_N$ are not distinct.
Similarly, by the biorthogonality relation \eqref{eq:biorth} and by \eqref{eq:BiOL}, we have that
$$\int_{\mathbb R}L_N(x,x;\vec a)\d x=N$$ for $a_1,\ldots, a_N$ distinct, and by continuity this continues to hold
in general.
\end{proof}

Our next result is a determinant expression for the marginals $\rho_{n,N}$ defined in \eqref{def:complexcorrelators}. It follows from a standard integrate-out trick, we record the proof here for the convenience of the reader, and because it is found in many classical references for Hermitian kernels only.

\begin{proposition}\label{prop:integratingout}
	The marginals $\rho_{n,N}$ of the measure $\mu_N$ defined by \eqref{def:BiOE-F}, satisfy the determinantal identity 
	\be
		\rho_{n,N}(x_1,\ldots, x_n)=
		\det\Big(L_N(x_j,x_k)\Big)_{j,k=1}^n, \quad n = 1, \ldots, N.
	\ee
\end{proposition}
\begin{proof}
We need to prove that
	\begin{equation}\label{eq:intout}
		\int_{\mathbb R^{N -n}} \det \Big(L_N(x_j,x_m)\Big)_{j,m = 1}^N \mathrm \d x_{n +1} \cdots \d x_{N} = (N - n)! 		\det\Big(L_N(x_j,x_k)\Big)_{j,k=1}^n,
	\end{equation}
	for $n = 1, \ldots, N -1$.
Let $k\in\{2,\ldots, N\}$, and write $M := \Big( L_N(x_i,x_j) \Big)_{i,j = 1}^k$. 
In the rest of this proof, we will write $M^{a,b}$ for the $(k-1)\times (k-1)$ matrix obtained from $M$ by removing the $a$-th line and the $b$-th column, and similarly
$M^{a_1 a_2, b_1 b_2}$ for the $(k-2)\times (k-2)$ matrix obtained from $M$ by removing the two corresponding rows and columns.
We now expand $\det M$ with respect to the last row, to obtain
\begin{align}
 	\int_{\mathbb R} \det M \d x_k &= \int_{\mathbb R} \sum_{j = 1}^k (-1)^{j + k} L_N(x_k,x_j) \det M^{k,j}\d x_k \nonumber \\ 
	&= \left(\int_{\mathbb R} L_N(x_k,x_k) \d x_k \right) \det M^{k,k} + \int_{\mathbb R}\sum_{j = 1}^{k-1} (-1)^{j + k} L_N(x_k,x_j) \det M^{k,j} \d x_k\nonumber.
	\end{align}
Observing that the integral in the first term is equal to $N$ (recall \eqref{eq:reproducing0}) and expanding $\det M^{k,j}$ once more, now with respect to the last column, we find
	\begin{align}
	\int_{\mathbb R} \det M \d x_k &= N \det M^{k,k} + \int_{\mathbb R} \sum_{j = 1}^{k-1} (-1)^{j + k} L_N(x_k,x_j) \sum_{m = 1}^{k-1} (-1)^{m + k - 1} L_N(x_m,x_k)\det M^{km,j k} \d x_k \nonumber\\
	&= N \det M^{k,k} - \sum_{j,m = 1}^{k-1} (-1)^{j + m} L_N(x_m,x_j) \det M^{km,j k}  \\&=  (N - k + 1) \det M^{k,k},
\end{align}
where we used the reproducing property \eqref{eq:reproducing0} in the second line, and where the last line is verified again by row or column expansion.
This implies that
$$
 \int_{\mathbb R} \det \left( L_N(x_i,x_j)\right)_{i,j = 1}^k \d x_k = (N - k + 1) \det \Big( L_N(x_i,x_j)\Big)_{i,j = 1}^{k-1},
$$
for every $k\in\{2,\ldots, N\}$. Repeatedly applying this identity for $k=N, N-1,\ldots, n+1$, we obtain \eqref{eq:intout}.
\end{proof}

Propositions \ref{prop:projection} and \ref{prop:integratingout} together imply part 1 of Theorem \ref{thm:biorth}.

	\medskip

If $a_1=\cdots=a_N=a$, we 
can still evaluate the $u$-integral in the double contour integral expression for $L_N$ using the residue theorem:
\[L_N(x,x')=\frac{1}{2\pi\i}\int_{\ell_N} {\rm Res}\left(\frac{\e^{ux'}}{W_0(u)(u-a)^N(v-u)};u=a\right)W_N(v)\e^{-vx}\d v.\]
The residue takes the form
$\e^{ax'}\sum_{j=1}^N c_j(v)x'^{j-1}$, with $c_j(v)=\frac{1}{v-a}P_{N-j}\left(\frac{1}{v-a}\right)$, for some polynomial $P_{N-j}$ of degree $N-j$. It follows that
\[L_N(x,x')=\e^{ax'}\sum_{j=1}^N x'^{j-1}\frac{1}{2\pi\i}\int_{\ell_N} W_N(v)c_j(v)\e^{-vx}\d v.\]
Hence,
 \eqref{def:BiOE-F} is equal to
	\begin{multline}\frac{1}{N!}\det\left(\e^{ax_k}\sum_{j=1}^N x_k^{j-1}\int_{\ell_N} W_N(v)c_j(v)\e^{-vx_m}\d v\right)_{k,m=1}^N\prod_{k=1}^N\d x_k\\
	=\frac{1}{N!}\det\left(x_k^{j-1}\right)_{k,j=1}^N\det\left(\frac{1}{2\pi\i}\int_{\ell_N} c_j(v)W_N(v)\e^{-vx_m}\d v\right)_{j,m=1}^N\prod_{k=1}^N \e^{a x_k} \d x_k.
	\end{multline}
In the latter expression, we can replace 	
	$\displaystyle\int_{\ell_N} c_j(v)W_N(v)\e^{-vx_m}\d v$, $j=1,\ldots, N$, by any set of functions with the same span, like for instance
	$\displaystyle\int_{\ell_N} \frac{(v-a)^{j-1}W_N(v)}{(v-a)^{N}}\e^{-vx_m}\d v$, $j=1,\ldots, N$, at the price of replacing the normalization constant $N!$ by a different constant $Z_N$. 
	We then obtain that
 
	\begin{align*}\d\mu_N(\vec x)&=\frac{1}{N!}\det\left(x_k^{j-1}\right)_{k,j=1}^N\det\left(\frac{1}{2\pi\i}\int_{\ell_N} c_j(v)W_N(v)\e^{-vx_k}\d v\right)_{j,k=1}^N\prod_{k=1}^N \e^{a x_k} \d x_k\\
	&=
\frac{1}{Z_N}\det\left(x_k^{j-1}\right)_{k,j=1}^N\det\left(\frac{1}{2\pi\i}\int_{\ell_N} \frac{(v-a)^{j-1}W_N(v)}{(v-a)^{N}}\e^{-vx_k}\d v\right)_{j,k=1}^N\prod_{k=1}^N \e^{a x_k} \d x_k	
	\\
	&=
\frac{1}{Z_N}\det\left(x_k^{j-1}\right)_{k,j=1}^N\det\left(\phi_j(x_k)\right)_{j,k=1}^N\prod_{k=1}^N  \d x_k
	\end{align*}

	 This proves part 3 of Theorem \ref{thm:biorth}.		
		
		\medskip
		
In a similar way, one can obtain explicit expressions for the biorthogonal measures in cases where some, but not all, of the zeros $a_1,\ldots, a_N$ are confluent.

\section{Determinant identities}\label{section:det}

In this section, we will prove the determinant identities from Theorem \ref{thm:Fredholm}. We will start by expressing $\mu_N[\sigma]$ as an $N\times N$ determinant, and afterwards we proceed towards the different Fredholm determinant identities. 
Again, throughout the section, we assume that ${\Sigma_N}, \ell_N, W_N$ satisfy Assumption \ref{assumptions}.

\subsection{Matrix determinants}\label{subsec:finitedet}

The finite size determinants of part 1 of Theorem \ref{thm:Fredholm} are obtained after a straightforward application of Andr\'eief's identity.
\begin{proposition}Let $\sigma:\mathbb R\to\mathbb C$ be bounded and measurable.
\begin{enumerate}\item If $a_1,\ldots, a_N$ are distinct, we have
\[\mu_N[\sigma]=\det\left(\int_{\mathbb R}(1-\sigma(x))\e^{a_m x}\psi_k(\e^x)\d x\right)_{m,k=1}^N.\]
\item If $a_1,\ldots, a_N=a$, we have
\[\mu_N[\sigma]=\frac{\displaystyle\det\left(\int_{\mathbb R}(1-\sigma(x))x^{m-1}\phi_k(\e^x)\d x\right)_{m,k=1}^N}{\displaystyle\det\left(\int_{\mathbb R}x^{m-1}\phi_k(\e^x)\d x\right)_{m,k=1}^N}.\]
\end{enumerate}
\end{proposition}
\begin{proof}
Let $a_1,\ldots, a_N$ be distinct. From the expression \eqref{BiOE-distinct} for $\mu_N$, we have
\[\mu_N[\sigma]=\frac{1}{N!}\int_{\mathbb R^N}\det\left(\e^{a_jx_k}\right)_{k,j=1}^N\det\left(\psi_j(\e^{x_m})\right)_{j,m=1}^N\prod_{j = 1}^N \left(1 - \sigma(x_j) \right)\d x_j.
\]
We now apply Andr\'eief's identity on the right hand side and obtain the result.

\medskip

If $a_1=\cdots=a_N=a$,
we have
\[\mu_N[\sigma]=\frac{1}{Z_N}\int_{\mathbb R^N}\det\left(x_j^{k-1}\right)_{k,j=1}^N\det\left(\phi_j(\e^{x_m})\right)_{j,m=1}^N\prod_{j = 1}^N \left(1 - \sigma(x_j) \right)\d x_j,
\]
with \[Z_N=\int_{\mathbb R^N}\det\left(x_j^{k-1}\right)_{k,j=1}^N\det\left(\phi_j(\e^{x_m})\right)_{j,m=1}^N\prod_{j=1}^N\d x_j.\]
Applying Andr\'eief's identity to the above two expressions, we obtain the result.
\end{proof}

\subsection{Fredholm determinants}

The following lemma is another classical result concerning biorthogonal ensembles. Once again, we recall the proof for convenience, and to convince the reader that it does not require positivity of the biorthogonal measure.

\begin{proposition}\label{prop:Fredholm1}
The average multiplicative statistics defined in \eqref{def:avgmultstat} satisfy the identity
\begin{equation}\label{eq:Prop3.2}
	\mu_N[\sigma]=1 + \sum_{n = 1}^N \frac{(-1)^n}{n!}\int_{\mathbb R^n} \det \Big( \sigma(x_j)L_N(x_j,x_m) \Big)_{j,m = 1}^n \d x_1\cdots \d x_n.
\end{equation}
\end{proposition}
\begin{proof}
We expand the left hand side of the identity as 
\begin{align*}
	\mu_N[\sigma]&=\int \prod_{j = 1}^N \left(1 - \sigma(x_j) \right) \d \mu_N(x_1,\ldots,x_N) 
	\\&= 1 + \sum_{n=1}^N (-1)^n \binom{N}{n} \int_{\mathbb R^N} \prod_{j = 1}^n \sigma(x_j) \d \mu_N(x_1,\ldots,x_N)  \nonumber \\
	&= 1 + \sum_{n = 1}^N (-1)^n \frac{1}{n!(N-n)!} \int_{\mathbb R^N} \prod_{j = 1}^n \sigma(x_j)\det\Big( L_N(x_k,x_{m})\Big)_{k,m = 1}^N \d x_1 \cdots \d x_N.
\end{align*}
Integrating out the variables $x_{n+1}, \ldots, x_N$, using Proposition \ref{prop:integratingout} (or the equivalent identity \eqref{eq:intout}), we obtain that the latter is equal to
\begin{gather}
1 + \sum_{n = 1}^N \frac{(-1)^n}{n!}\int_{\mathbb R^n} \det \left( \sigma(x_j)L_N(x_j,x_m) \right)_{j,m = 1}^n \d x_1\cdots \d x_n, \label{eq:5.8}
\end{gather}
\end{proof}

We recognize the right hand side of the identity \eqref{eq:Prop3.2} as the truncated Fredholm series $\det(1-\sigma L_N)_{L^2(\mathbb R)}$, and the identity \eqref{eq:Fredholm1} is proven.

\medskip

To prove the other identities in Theorem \ref{thm:Fredholm} (points 3 and 4), we need to rely on properties of Fredholm determinants. For that reason, we will now construct trace-class operators associated to suitable conjugations of the kernels we considered before, such that we can interpret the truncated Fredholm series as proper Fredholm determinants associated to (finite rank) trace-class operators.

\medskip

The kernel \[L_N(x,x')=\sum_{m=1}^N\e^{a_mx'}\psi_m(\e^{x})\] does not directly define an integral operator on $L^2(\mathbb R)$, because neither $\e^{a_m x'}$ nor $\psi_m(\e^x)$ are (in general) in $L^2(\mathbb R)$. We can however define a new kernel of the form
\be\label{def:hatL}\widetilde L_N(x,x'):=\frac{h(x)}{h(x')}L_N(x,x'),\ee
such that the Fredhom series $\det(1-\sigma L_N)_{L^2(\mathbb R)}$ and $\det(1-\sigma \widetilde L_N)_{L^2(\mathbb R)}$ are equal by definition, but the integral operator $\widetilde L_N$ defined by
\be\label{def:opL}\widetilde L_N[f](x):=\int_{\mathbb R}\widetilde L_N(x,x')f(x')\d x',\ee
is a genuine integral operator acting on $L^2(\mathbb R)$, of rank $N$ and in particular trace-class.
A suitable choice for $h$ turns out to be
\be\label{def:h}
h(x)=\begin{cases}
\e^{(a_{\max}+\epsilon) x},& x\geq 0,\\
\e^{(a_{\min}-\epsilon) x},& x<0,
\end{cases}
\ee
for sufficiently small $\epsilon>0$, with $a_{\min}$ and $a_{\max}$ as in \eqref{def:aminmax}. Then, the asymptotic behavior \eqref{eq:aspsi} for $\psi(\e^x)$ as $x\to \pm\infty$ implies indeed that $\widetilde L_N$ is a bounded linear operator on $L^2(\mathbb R)$.
Similarly, we define for $H_N^\sigma$ from \eqref{def:H} 
the conjugated kernel
\be\label{def:hatH}
\widetilde H_N^\sigma(y,y'):=\e^{-(a_{\max}+\epsilon) (y-y')}
H_N^\sigma(y,y'),
\ee
and the associated integral operator
\[{\widetilde H_N^\sigma}[f](y)=\int_{0}^\infty\widetilde H_N^\sigma(y,y')f(y')\d y',\]
acting on $L^2(0,\infty)$.

\begin{proposition}\label{prop:HL}
Let $W_N$ satisfy Assumption \ref{assumptions} and be such that
$W_N(v)=O(|v|^{-1-\epsilon})$ as $v\to \infty$ in $\mathcal S_N$.
Let $\sigma:\mathbb R\to\mathbb C$ be bounded and measurable, and such that there exists $\epsilon>0$ such that
\[\sigma(x)=O\left(\e^{(a_{\max}-a_{\min}+\epsilon)x}\right),\qquad x\to -\infty.\]
The operators ${\widetilde H_N^\sigma}$ and $\sigma\widetilde L_N$ are trace-class, and we have the identity
\begin{equation}
		\det(1 - {\widetilde H_N^\sigma})_{L^2(0,\infty)} = \det(1 - \sigma \widetilde L_N)_{L^2(\mathbb R)}.
	\end{equation}
\end{proposition}
\begin{proof}
Given $\Psi_1$ and $\Psi_2$  as in \eqref{def:psi12}, define two operators $$ {A: L^2(\mathbb R) \to L^2(0,\infty) \quad \text{and} \quad B: L^2(0,\infty) \to L^2(\mathbb R)}$$ with kernels
	$$
	 A(y,x) := \frac{\e^{-(a_{\max}+\epsilon) y}}{h(x)}\Psi_1\left({\rm e}^{y + x} \right), \quad B(x,y') = \e^{(a_{\max}+\epsilon) y'}h(x)\sigma(x)\Psi_2\left({\rm e}^{y' + x} \right).
	$$
	Observe also that $\Psi_2$ is well-defined since $W_N(v)=O(|v|^{-1-\epsilon})$ as $v\to\infty$ in $\mathcal S_N$.
Using Assumption \ref{assumptions} and the definition \eqref{def:psi12} of $\Psi_1$, it is straightforward to check by deforming $\Sigma_N$ to a sufficiently narrow curve encircling $a_1,\ldots, a_N$ that, for any $\epsilon>0$,
\[\Psi_1(\e^{y+x})=O\left(\e^{(a_{\min}-\epsilon/2)(y+x)}\right)\ \mbox{as $y+x\to -\infty$,}\quad \Psi_1(\e^{y+x})=O\left(\e^{(a_{\max}+\epsilon/2)(y+x)}\right)\ \mbox{as $y+x\to +\infty$},\]
such that
\[A(y,x)=O\left(\e^{-\epsilon|x+y|/2}\right),\qquad x+y\to \pm\infty.\]
It is now straightforward to check that ${A}$ is a Hilbert-Schmidt operator.

\medskip

Similarly, for sufficiently small $\epsilon$, we can deform $\ell_N$ to the vertical line with real part $a_{\max}+\frac{3}{2}\epsilon$, and hence
\[\Psi_2(\e^{y+x})=O\left(\e^{-(a_{\max}+\frac{3}{2}\epsilon)(y+x)}\right)\quad \mbox{as $y+x\to \pm\infty$},\]
such that
\[B(x,y)=O\left(\e^{-\frac{\epsilon}{2} y}\e^{-(a_{\max}+\frac{3}{2}\epsilon) x}\sigma(x)h(x)\right),\qquad x+y\to \pm\infty,\]
which implies that ${B}$ is also a Hilbert-Schmidt operator, under the condition that
\[\sigma(x)=\e^{(a_{\max}-a_{\min}+\epsilon)x},\qquad x\to -\infty.\]

As compositions of Hilbert-Schmidt operators, {$ A B$ and $ B A$} are trace-class operators, and we have
{\be \det(1- A B)_{L^2(0,\infty)}=\det(1- B A)_{L^2(\mathbb R)}.\ee }
It is easy to see that {$\widetilde H_N^\sigma= A B$ by comparing the kernels of both operators, and we will now prove that $\sigma\widetilde L_N= B A$}, which implies the result.

\medskip

Using the integral representation of $\Psi_1$ and $\Psi_2$ from \eqref{def:psi12}, we compute the kernel of the operator ${B A}$ acting on $L^2(\mathbb R)$ as follows:
	\begin{align*}
		(BA)(x,x')& = \frac{1}{(2 \pi \i)^2}\sigma(x)\frac{h(x)}{h(x')} \int_0^\infty \d y \int_{{\Sigma_N}} \d u \frac{{\rm e}^{(y + x')u}}{W_N(u)} \int_{\ell_N} \d v W_N(v){\rm e}^{-(y + x)v} \nonumber \\
		&=  \frac{1}{(2 \pi \i)^2}\sigma(x)\frac{h(x)}{h(x')}  \int_{{\Sigma_N}} \d u \frac{{\rm e}^{ x'u}}{W_N(u)} \int_{\ell_N} \d v W_N(v){\rm e}^{-xv}\int_0^\infty \d y {\rm e}^{y(u-v)} \nonumber \\
		&= \sigma(x)\frac{h(x)}{h(x')}L_N(x,x')=\sigma(x)\widetilde L_N(x,x'), \label{eq:4.10}
	\end{align*}
	such that indeed {$B A=\sigma\widetilde L_N$}.
\end{proof}
\begin{remark}
If $\sigma$ is differentiable, we can use integration by parts to rewrite $H_N^{\sigma}(y,y')$ as
\begin{align*}
H_N^{\sigma}(y,y')&=\int_{\mathbb R} \d r \sigma(r)\int_{{\Sigma_N}} \frac{\d u}{2 \pi \i} \frac{1}{W_N(u)} \e^{(y+r)u} \int_{\ell_N} \frac{\d v}{2 \pi \i} W_N(v)\e^{-(y'+r)v}\\
&=\frac{1}{(2\pi\i)^2}\int_{\R} \d r \sigma'(r) \int_{{\Sigma_N}}\d u \int_{\ell_N}\d v \frac{W_N(v)}{W_N(u)} \e^{(y+r)u} \e^{-(y'+r)v}\frac{1}{v-u}
\\
&=\int_{\R} \d r \sigma'(r) L_N\left(y'+r, y+r\right).
\end{align*}
Fredholm determinants with kernels of this type can be useful to obtain underlying differential equations and integro-differential equations, like for instance in \cite{BCT}.
\end{remark}
We constructed { the Hilbert-Schmidt operators $\widetilde L_N$ and $\widetilde H_N^\sigma$} in such a way that
{\[\det(1 - \sigma\widetilde L_N)_{L^2(\mathbb R)} = \det(1 - \sigma  L_N)_{L^2(\mathbb R)},\qquad \det(1 - \widetilde H_N^\sigma)_{L^2(0,\infty)} = \det(1 - H_N^\sigma)_{L^2(0,\infty)},\]}
where with the left hand sides, we mean proper Fredholm determinants, and with the right hand sides, we mean Fredholm series. As a consequence, combining Proposition \ref{prop:HL} with Proposition \ref{prop:Fredholm1}, we obtain \eqref{eq:Fredholm2}.

The result below generalizes 
Lemma 8.8 in \cite{BCF}, and will imply \eqref{eq:Fredholm3}.
\begin{proposition}\label{lemma:alternativeFredholmdet}
Let $W_N$ satisfy Assumption \ref{assumptions} and be such that
$W_N(v)=O(|v|^{-1-\epsilon})$ as $v\to \infty$ in $\mathcal S_N$.
Let $t\in\mathbb R$, and let $\sigma_t(x)=\frac{1}{1+\e^{-x-t}}$.
If $a_{\max}-a_{\min}<1$,
we have the identity
\be
{\det\left(1-\widetilde H_N^{\sigma_t}\right)_{L^2(\mathbb R)}=\det\left(1+ K_{N,t}\right)_{L^2({\Sigma_N})}},
\ee
with {$ K_{N,t}$} the integral operator acting on $L^2(\Sigma_N)$ with kernel $K_{N,t}$ defined in \eqref{def:kernel}.
\end{proposition}

\proof

Using the fact that $\frac{1}{v - u'} = \displaystyle\int_{0}^\infty \d y\, {\rm e}^{-y(v - u')}$ when $\mathrm{Re}(v - u') > 0$, we can write $K_{N,t}$ from \eqref{def:kernel} as the composition of two Hilbert-Schmidt operators, {$ K_{N,t} =  C D$}, with
$${ C : L^2(0,\infty) \to L^2({\Sigma_N})\; \text{and} \;  D: L^2({\Sigma_N}) \to L^2(0,\infty)}$$ with kernels
$$
	C(u,y) = \frac{\e^{yc}}{2 \pi \i} \int_{\ell_N} \frac{\pi \e^{t(v-u)}}{\sin \pi(u - v)} \frac{W_N(v)}{W_N(u)} \e^{-yv} \d v, \quad \quad D(y,u') = \frac{1}{2\pi\i}{\rm e}^{y(u'-c)},
$$
where $\ell_N$ is the vertical line with real part $c$, with $a_{\max}<c<a_{\min}+1$, and where $\Sigma_N$ lies at the left of $\ell_N$.
It is then straightforward to verify that ${D}$ is a Hilbert-Schmidt operator.
For ${C}$, the same is true, but note that we need here that $0<\Re(v-u)<1$ for every $v\in\ell_N$, $u\in{\Sigma_N}$ in order to avoid zeros of $\sin\pi(u-v)$.

\medskip

Now we use the fact that {$\det(1 +  C D)_{L^2({\Sigma_N})} = \det(1 + D C)_{L^2(0,\infty)}$, and compute the kernel of $ D C$} as (note that the minus sign in front compensates the modified sign of the argument of $\sin$)
$$
(DC)(y,y') = -\frac{1}{(2 \pi \i)^2} \int_{{\Sigma_N}} \d u \ {\rm e}^{y(u-c)} \int_{\ell_N} \d v \frac{\pi \e^{t(v-u)}}{\sin \pi (v - u)} \frac{W_N(v)}{W_N(u)} {\rm e}^{-y'(v-c)}.
$$
We then substitute the  identity
$$
\frac{\pi}{\sin \pi s} = \int_{\mathbb R} \frac{{\rm e}^{-sx}}{1 + {\rm e}^{-x}}\d x,\qquad 0<\Re s<1,
$$ 
where we use once more  that  $0<\Re(v-u)<1$ for every $v\in\ell_N$, $u\in{\Sigma_N}$,
and change the order of integration, to obtain
\begin{align*}
(DC)(y,y') &= - {\rm e}^{c(y' - y)}\int_{\R} \frac{\d x}{1 + {\rm e}^{-x}} \int_{{\Sigma_N}} \frac{\d u}{2 \pi \i} \frac{1}{W_N(u)}\e^{(y+x-t)u}\int_{\ell_N} \frac{\d v}{2 \pi \i} W_N(v)\e^{-(y'+x-t)v}\nonumber \\
&= -{\rm e}^{c(y' - y)} \int_{\R} \frac{\d x}{1 + {\rm e}^{-x}} \int_{{\Sigma_N}} \frac{\d u}{2 \pi \i} \frac{1}{W_N(u)} \e^{(y+x-t)u} \int_{\ell_N} \frac{\d v}{2 \pi \i} W_N(v)\e^{-(y'+x-t)v}\\
&= -{\rm e}^{c(y' - y)} \int_{\R} \frac{\d r}{1 + {\rm e}^{-r-t}} \int_{{\Sigma_N}} \frac{\d u}{2 \pi \i} \frac{1}{W_N(u)} \e^{(y+r)u} \int_{\ell_N} \frac{\d v}{2 \pi \i} W_N(v)\e^{-(y'+r)v}
, \label{eq:beforedoublecontour}
\end{align*}
which is exactly $-\widetilde H_N^{\sigma_t}(y,y')$ by \eqref{def:H} and \eqref{def:psi12}, if $c=a_{\max}+\epsilon$.
\endproof

\section{Logarithmic derivatives and deformed biorthogonal measures}\label{section:deformations}

In this section, we will prove Theorem \ref{thm: deformed}. The proof of this result is essentially a special case of a more general result in \cite{ClaeysGlesner}. Besides Assumption \ref{assumptions}, we assume here that $\sigma:\mathbb R\to\mathbb C$ is $C^1$, bounded, and with bounded derivative, and that $t$ is such that
$$\mu_N[\sigma_t]=\det(1-\sigma_t L_N)_{L^2(\mathbb R)} = \det(1-\sigma_t \widetilde L_N)_{L^2(\mathbb R)}\neq 0,$$ 
with $\sigma_t(x)=\sigma(x+t)$.

\medskip

For simplicity, we will first assume that $a_1,\ldots, a_N$ are distinct. Later, we will use a continuity argument to prove our identities in general.
By Proposition \ref{prop:biorth}, we know that $\widetilde L_N$ can be written in the form
\be\label{def:hatL2}
\widetilde L_N(x;x')=\sum_{k=1}^NF_k(x)G_k(x'),\qquad \int_{\mathbb R}F_k(x)G_j(x)\d x=\delta_{jk}, \quad \forall\, j,k = 1,\ldots,N,
\ee
with $F_k(x)=h(x)\psi_k(\e^x)$ and $G_k(x)=\e^{a_k x}/h(x)$.
The asymptotic behavior \eqref{eq:aspsi} of $\psi_k$ and the definition \eqref{def:h} imply that
$F_k(x)$ and $G_k(x)$ decay rapidly as $x\to\pm\infty$. Consequently, the operator $\widetilde L_N$ is a bounded linear operator, which is moreover finite rank and thus trace-class.
By the Jacobi identity, we have
\begin{equation}
\frac{\d}{\d t} \log \mu_N[\sigma_t]=\frac{\d}{\d t} \log \det(1-\sigma_t\widetilde L_N)_{L^2(\mathbb R)} = -\Tr \left( \left(\frac{\d}{\d t} \sigma_t\right)\widetilde L_N(1 - \sigma_t\widetilde L_N)^{-1}  \right)_{L^2(\mathbb R)},
\end{equation}
where we observe that $(\frac{\d}{\d t} \sigma_t)\widetilde L_N(1 - \sigma_t\widetilde L_N)^{-1}=\sigma_t'\widetilde L_N(1 - \sigma_t\widetilde L_N)^{-1} $ is a trace-class operator if $\sigma_t'$ is bounded.
We thus have
\begin{equation}\label{eq:logder4}
\frac{\d}{\d t} \log \mu_N[\sigma_t] = -\Tr \left( \sigma_t'\widetilde L_N(1 - \sigma_t\widetilde L_N)^{-1}  \right)_{L^2(\mathbb R)}. 
\end{equation}

\medskip

By definition \eqref{def:opL} and by Proposition \ref{prop:projection}, we directly verify that $\widetilde L_N^2=\widetilde L_N$, such that $\widetilde L_N$ is a rank $N$ (not necessarily Hermitian) projection operator. 
We denote $\mathcal F_N$ for the linear span of $F_1,\ldots, F_N$ and $\mathcal G_N$ for the linear span of $G_1,\ldots, G_N$.
\begin{proposition}
$\widetilde L_N$ is the unique projection operator on $L^2(\mathbb R)$ with range $\mathcal F_N$ and kernel $\mathcal G_N^{\perp}$.
\end{proposition}
\begin{proof}
It is straightforward to verify from \eqref{def:hatL}--\eqref{def:opL} and \eqref{def:hatL2}
that $\widetilde L_N[f]=f$ for $f\in \mathcal F_N$ and that $\widetilde L_N[f]=0$ for $f\in \mathcal G_N^{\perp}$.
Moreover, $\mathcal F_N\oplus \mathcal G_N^{\perp}=L^2(\mathbb R)$. Indeed, suppose that this is not the case; then there exists a non-zero $f\in L^2(\mathbb R)$ such that $f\perp h$ for every $h\in  \mathcal F_N\oplus  \mathcal G_N^{\perp}$.
Consequently, $f\in \mathcal F_N^{\perp}\cap \mathcal G_N$. but 
$\mathcal F_N^{\perp}\cap \mathcal G_N=\{0\}$, 
by the orthogonality relation $\int_{\mathbb R}F_j(x)G_k(x)\d x=\delta_{jk}$, which gives a contradiction.
Using the above properties, we obtain that $\widetilde L_N$ is the unique projection with range $\mathcal F_N$ and kernel $\mathcal G_N^{\perp}$.
\end{proof}
Let us now consider a kernel $M_t(x,x')$ of the operator $\widetilde L_N(1-\sigma_t \widetilde L_N)^{-1}$ acting on $L^2(\mathbb R)$, such that
\[\widetilde L_N(1-\sigma_t \widetilde L_N)^{-1}[f](x)=\int_{\mathbb R}M_t(x,x')f(x')\d x'.\]
Then we define an operator $\mathcal P_t$ acting on $L^2\left(\mathbb R,(1-\sigma_t(x))\d x\right)$, with respect to the deformed Lebesgue measure $(1-\sigma_t(x))\d x$, as follows:
\be\label{def:Pt}\mathcal P_t[g](x)=\widetilde L_N(1-\sigma_t \widetilde L_N)^{-1}[(1-\sigma_t)g]=\int_{\mathbb R}M_t(x,x')g(x')(1-\sigma_t(x'))\d x'.\ee

\begin{proposition}The operator
$\mathcal P_t$
is the unique projection operator on $L^2\left((1-\sigma_t(x))\d x\right)$ with range $\mathcal F_N$ and kernel $\mathcal G_N^{\perp}$.
\end{proposition}
\begin{proof}
Note first that the orthogonal complement in $\mathcal G_N^{\perp}$ here stands for orthogonality with respect to the deformed Lebesgue measure $(1-\sigma_t(x))\d x$. 

\medskip

Let us first prove that $\mathcal F_N\oplus \mathcal G_N^{\perp}=L^2\left((1-\sigma_t(x))\d x\right)$.
Similarly as before, suppose that this is not true, then there is a non-zero $g\in L^2\left(\mathbb R,(1-\sigma_t(x))\d x\right)$ in the orthogonal complement of 
$\mathcal F_N\oplus \mathcal G_N^{\perp}$. This means that $g\in \mathcal F_N^{\perp}\cap \mathcal G_N$. 
In other words, $f=g(1-\sigma_t)\in L^2(\mathbb R,\d x)$ belongs to the $L^2(\mathbb R, \d x)$-orthogonal complement of $\mathcal F_N$, which has trivial intersection with $\mathcal G_N$. This gives a contradiction, and we conclude that $\mathcal F_N\oplus \mathcal G_N^{\perp}=L^2\left((1-\sigma_t(x))\d x\right)$.

\medskip

Let $g\in \mathcal F_N$, such that $(1-\sigma_t)g=(1-\sigma_t\widetilde L_N)g$.
Then,
\[\mathcal P_t[g]=\widetilde L_N(1-\sigma_t \widetilde L_N)^{-1}[(1-\sigma_t)g]=\widetilde L_N[g]=g.\]

\medskip

Finally, let $g\in \mathcal G_N^{\perp}$, such that $\widetilde L_N[(1-\sigma_t)g]=0$ and $(1-\sigma_t)g=\left(1-\sigma_t\widetilde L_N\right)[(1-\sigma_t)g]$.
Thus,
\[\mathcal P_t[g]=\widetilde L_N[(1-\sigma_t)g]=0.\]

\medskip

Combining the above properties, we obtain the result.
\end{proof}

Recall from \eqref{def:Pt} that the $L^2\left(\mathbb R,(1-\sigma_t(x))\d x\right)$-operator $\mathcal P_t$ has integral kernel $M_t(x,x')$, i.e.
\[\mathcal P_t[g]=\widetilde L_N(1-\sigma_t\widetilde L_N)^{-1}[(1-\sigma_t)g]=\int_{\mathbb R}M_t(x,x')g(x')(1-\sigma_t(x'))\d x,\]
and we can take $M_t$ of the form
\[M_t(x,x')=\sum_{m=1}^{N}\widetilde F_m(x)\widetilde G_m(x'), \quad \text{with} \quad \int_{\mathbb R}\widetilde F_j(x)\widetilde G_k(x)(1-\sigma_t(x))\d x=\delta_{jk}, \quad j,k = 1, \ldots,N,\] 
where $\widetilde F_1,\ldots, \widetilde F_N$ span the range $\mathcal F_N$ of $\widetilde L_N$, and $\widetilde G_1,\ldots, \widetilde G_N$ span $\mathcal G_N$, since $\mathcal G_N^\perp$ (in $L^2(\mathbb R,(1-\sigma_t(x))\d x)$) is the kernel of $\mathcal P_t$.
The biorthogonal measure
\[\d\widetilde\nu_{N,t}(x_1,\ldots, x_N):=\frac{1}{N!}\det\left(M_t(x_j,x_k)\right)_{j,k=1}^N\prod_{k=1}^N (1-\sigma_t(x_k))\d x_k\]
can then be rewritten as
\begin{align*}
\d\widetilde\nu_{N,t}(\vec x)&=\frac{1}{N!}\det\left(\sum_{m=1}^{N}\widetilde F_m(x_j)\widetilde G_m(x_k)\right)_{j,k=1}^N\prod_{k=1}^N (1-\sigma_t(x_k)\d x_k\\
&=\frac{1}{N!}\det\left(\widetilde F_m(x_j)\right)_{j,m=1}^N\det\left(\widetilde G_m(x_k)\right)_{m,k=1}^N\prod_{k=1}^N (1-\sigma_t(x_k))\d x_k
\\
&=\frac{1}{Z_{N,t}}\det\left(F_m(x_j)\right)_{j,m=1}^N\det\left(G_m(x_k)\right)_{m,k=1}^N\prod_{k=1}^N (1-\sigma_t(x_k))\d x_k\\
&=\frac{1}{N!}\det\left(\sum_{m=1}^{N}F_m(x_j)G_m(x_k)\right)_{j,k=1}^N\prod_{k=1}^N (1-\sigma_t(x_k))\d x_k\\
&=\frac{1}{Z_{N,t}}\det\left(L_N(x_j,x_k)\right)_{j,k=1}^N\prod_{k=1}^N (1-\sigma_t(x_k))\d x_k=\d\nu_{N,t}(\vec x).
\end{align*}
We can now conclude from \eqref{eq:logder} that 
\[\frac{\d}{\d t} \log \mu_N[\sigma_t] = -\Tr \left( \sigma_t'\widetilde L_N(1 - \sigma_t\widetilde L_N)^{-1}  \right)_{L^2(\mathbb R)}=\int_{\mathbb R}\sigma_t'(x)M_t(x,x)\d x, 
\]
and recalling \eqref{def:onepoint} and \eqref{def:onepoint2}, we moreover have
$M_t(x,x)=\widetilde\kappa_{N,t}(x)$, and $M_t(x,x)=\frac{\kappa_{N,t}(x)}{1-\sigma_t(x)}$ if $\sigma_t(x)\neq 1$. The above identities imply Theorem \ref{thm: deformed} in the case of distinct $a_1,\ldots, a_N$, and the general identities follow as before by continuity.

\section{Log Gamma polymer partition function and LUE with external source}\label{section:LogGamma}

In the last three sections of this paper, we give an overview of concrete polymer and random matrix models which fit in our general framework, and to which the results presented in Section \ref{sec:Intro} can be applied. A concise table with a (probably non-exhaustive) list of models which can be described in terms of our general class of biorthogonal measures is given at the end, see Table \ref{table}. 

\paragraph{The model.}
The Log Gamma polymer was introduced by Sepp\"al\"ainen in \cite{Seppalainen}.
Let $n,N>1$ be two positive integers.
We consider the lattice
\[\Big\{(j,k):j\in\{1,2,\ldots, n\}, k\in\{1,2,\ldots, N\}\Big\},\]
and assign to each point $(j,k)$ in the lattice independently a random variable $d_{j,k}\sim\Gamma^{-1}(\gamma_{j,k})$ with inverse Gamma distribution,
\[\mathbb P(d_{j,k}\leq y)=\frac{1}{\Gamma(\gamma_{j,k})}\int_{0}^y  x^{-\gamma_{j,k} - 1}{\rm e}^{-1/x}\d x.\]
The mean of $d_{j,k}$ thus exists only for $\gamma_{j,k}>1$ and is then equal to $\frac{1}{\gamma_{j,k}-1}$, while the variance is finite only if $\gamma_{j,k}>2$, and then equal to $\frac{1}{(\gamma_{j,k}-1)^2(\gamma_{j,k}-2)}$.
The parameters $\gamma_{j,k}$ take the form
$\gamma_{j,k}=\alpha_j-a_k$, where
$\vec\alpha=(\alpha_j)_{j = 1}^n$ and $\vec a=( a_k )_{k = 1}^N$ are two vectors of parameters such that $\alpha_j-a_k>0$ for all $j=1,\ldots, n$ and $k=1,\ldots, N$.

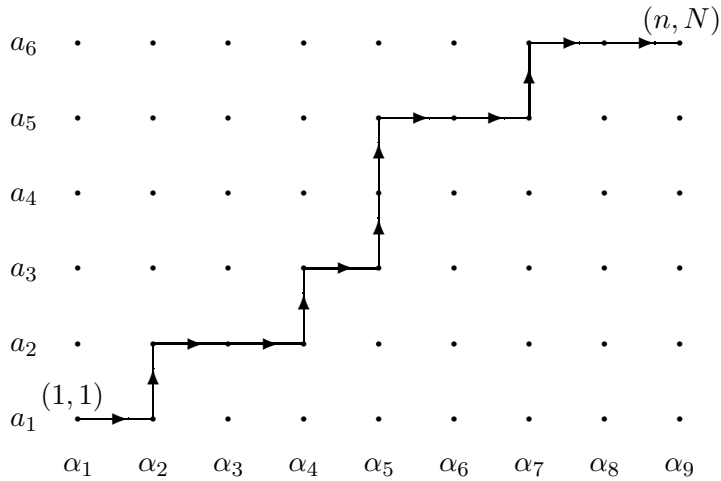
\begin{figure}[t]
\begin{center}
    \setlength{\unitlength}{1truemm}
    \begin{picture}(100,50)(0,-10)
    \put(5,2){$(1,1)$} \put(85,52){$(n,N)$}
    \put(10,0){\thicklines\circle*{.8}}
\put(20,0){\thicklines\circle*{.8}}
\put(30,0){\thicklines\circle*{.8}}
\put(40,0){\thicklines\circle*{.8}}
\put(50,0){\thicklines\circle*{.8}}
\put(60,0){\thicklines\circle*{.8}}    
    \put(70,0){\thicklines\circle*{.8}}    
   \put(80,0){\thicklines\circle*{.8}}     
    \put(90,0){\thicklines\circle*{.8}}    
    
  \put(10,10){\thicklines\circle*{.8}}
\put(20,10){\thicklines\circle*{.8}}
\put(30,10){\thicklines\circle*{.8}}
\put(40,10){\thicklines\circle*{.8}}
\put(50,10){\thicklines\circle*{.8}}
\put(60,10){\thicklines\circle*{.8}}    
    \put(70,10){\thicklines\circle*{.8}}    
   \put(80,10){\thicklines\circle*{.8}}     
    \put(90,10){\thicklines\circle*{.8}}

    \put(10,20){\thicklines\circle*{.8}}
\put(20,20){\thicklines\circle*{.8}}
\put(30,20){\thicklines\circle*{.8}}
\put(40,20){\thicklines\circle*{.8}}
\put(50,20){\thicklines\circle*{.8}}
\put(60,20){\thicklines\circle*{.8}}    
    \put(70,20){\thicklines\circle*{.8}}    
   \put(80,20){\thicklines\circle*{.8}}     
    \put(90,20){\thicklines\circle*{.8}}  
    
    \put(10,30){\thicklines\circle*{.8}}
\put(20,30){\thicklines\circle*{.8}}
\put(30,30){\thicklines\circle*{.8}}
\put(40,30){\thicklines\circle*{.8}}
\put(50,30){\thicklines\circle*{.8}}
\put(60,30){\thicklines\circle*{.8}}    
    \put(70,30){\thicklines\circle*{.8}}    
   \put(80,30){\thicklines\circle*{.8}}     
    \put(90,30){\thicklines\circle*{.8}}  
    
\put(10,40){\thicklines\circle*{.8}}
\put(20,40){\thicklines\circle*{.8}}
\put(30,40){\thicklines\circle*{.8}}
\put(40,40){\thicklines\circle*{.8}}
\put(50,40){\thicklines\circle*{.8}}
\put(60,40){\thicklines\circle*{.8}}    
    \put(70,40){\thicklines\circle*{.8}}    
   \put(80,40){\thicklines\circle*{.8}}     
    \put(90,40){\thicklines\circle*{.8}}      
    
    \put(10,50){\thicklines\circle*{.8}}
\put(20,50){\thicklines\circle*{.8}}
\put(30,50){\thicklines\circle*{.8}}
\put(40,50){\thicklines\circle*{.8}}
\put(50,50){\thicklines\circle*{.8}}
\put(60,50){\thicklines\circle*{.8}}    
    \put(70,50){\thicklines\circle*{.8}}    
   \put(80,50){\thicklines\circle*{.8}}     
    \put(90,50){\thicklines\circle*{.8}}

    \put(10,0){\line(1,0){10}}\put(16.7,0){\thicklines\vector(1,0){.0001}}
    \put(20,10){\line(1,0){10}}\put(26.7,10){\thicklines\vector(1,0){.0001}}
    \put(30,10){\line(1,0){10}}\put(36.7,10){\thicklines\vector(1,0){.0001}}
    \put(40,20){\line(1,0){10}}\put(46.7,20){\thicklines\vector(1,0){.0001}}
    \put(50,40){\line(1,0){10}}\put(56.7,40){\thicklines\vector(1,0){.0001}}
    \put(60,40){\line(1,0){10}}\put(66.7,40){\thicklines\vector(1,0){.0001}}
    \put(70,50){\line(1,0){10}}  \put(76.7,50){\thicklines\vector(1,0){.0001}}  
    \put(80,50){\line(1,0){10}} \put(86.7,50){\thicklines\vector(1,0){.0001}}
    
    \put(20,0){\line(0,1){10}}\put(20,6.7){\thicklines\vector(0,1){.0001}}
    \put(40,10){\line(0,1){10}}\put(40,16.7){\thicklines\vector(0,1){.0001}}
    \put(50,20){\line(0,1){10}}\put(50,26.7){\thicklines\vector(0,1){.0001}}
    \put(50,30){\line(0,1){10}}\put(50,36.7){\thicklines\vector(0,1){.0001}}
    \put(70,40){\line(0,1){10}}\put(70,46.7){\thicklines\vector(0,1){.0001}} 
    
    \put(8,-7){$\alpha_1$}
    \put(18,-7){$\alpha_2$}
    \put(28,-7){$\alpha_3$}
    \put(38,-7){$\alpha_4$}
    \put(48,-7){$\alpha_5$}
    \put(58,-7){$\alpha_6$}
    \put(68,-7){$\alpha_7$}
    \put(78,-7){$\alpha_8$}
    \put(88,-7){$\alpha_9$}
    
    \put(1,-1){$a_1$}
    \put(1,9){$a_2$}
    \put(1,19){$a_3$}
    \put(1,29){$a_4$}
    \put(1,39){$a_5$}
    \put(1,49){$a_6$}

%

    
    \end{picture}
    \caption{The Log Gamma polymer lattice for $n=9$, $N=6$, with a possible up-right path. The values $\alpha_j, a_k$ determine the weight of the point $(j,k)$ in an up-right path.}
    \label{figure: LogGamma}
\end{center}
\end{figure}

	The Log Gamma polymer partition function $Z_{n,N}^{{\rm Log}\Gamma}({\vec \alpha, \vec a})$ is the random variable
	\begin{equation}
		Z_{n,N}^{{\rm Log}\Gamma}(\vec \alpha, \vec a) := \sum_{\pi: (1,1) \nearrow (n,N)} \prod_{(j,k) \in \pi} d_{j,k},
	\end{equation}
	where the sum is taken over all up-right directed lattice paths $\pi$ from $(1,1)$ to $(n,N)$, and the product over all lattice points in an up-right directed lattice path. See Figure \ref{figure: LogGamma}.

Heuristically, given the fact that the mean and variance of the random variables $d_{j,k}$ are small for $\gamma_{j,k}$ large, we can expect that for large values of $\gamma_{j,k}$ , the weights $\prod_{(j,k) \in \pi} d_{j,k}$ of different paths $\pi$ will then typically be small and not show large deviations, such that many different paths will contribute to the leading order of the partition function (large entropy). For $\gamma_{j,k}$ small, we expect on the other hand large fluctuations, and then we expect that the partition function will be dominated by few paths with large weights (small entropy), such that the model degenerates to a last passage percolation or corner growth model. See \cite{Seppalainen, BCR} for more details about this model.

\paragraph{Fredholm determinant identity.}
Suppose that $n\geq N$ and that there exist constants $\delta_1,\delta_2\in (0,1)$ such that $\delta_1 < \min \{\delta_2,1 - \delta_2 \} < 1/2,$ and such that
\begin{equation}\label{eq:assumptionLogG}
	0 \leq a_j < \delta_1, \quad \alpha_k > \delta_2, \quad \forall j = 1,\ldots,N, \; \forall k =1,\ldots, n.
\end{equation}
Under these assumptions,
Borodin, Corwin, and Remenik in \cite[Corollary 1.8]{BCR} proved the identity	
	\begin{equation}\label{eq:LaplaceLogGamma}
		\mathbb E \left[ {\rm e}^{-\e^t Z_{n,N}^{{\rm Log}\Gamma}(\valpha,\va)}\right] = \det(1 + K_{N,t})_{L^2({\Sigma_N})},
	\end{equation}
where $K_{N,t}$ {(now also depending on $n, \vec\alpha,\va$)} is as in \eqref{def:kernel}, with ${\Sigma_N}$ the circle of radius $\delta_1$ around the origin, $\ell_N$ the vertical line $\Re z=\delta_2$, and $W_N$ given by
\be\label{def:FLogGamma}
W_N(z)=W_{n,N}^{{\rm Log}\,\Gamma}(z;\vec\alpha,\vec a):=\frac{\prod_{j=1}^n\Gamma(\alpha_j-z)}{\prod_{k=1}^N\Gamma(z-a_k)}.
\ee
Assumption \ref{assumptions} is valid for this choice of $W_N$ if $n>N$, and also for $n=N$ if $\sum_{k=1}^N(\alpha_j+a_j-2\delta_2)<0$. The stronger decay assumption for $W_N$ needed in parts 3 and 4 of Theorem \ref{thm:Fredholm} is also valid if 
$n>N$, and if $n=N$ provided that $\sum_{k=1}^N(\alpha_j+a_j-2\delta_2)<-1$.
This is easily seen by using the following variant of Stirling's approximation (see, e.g., \cite[formula (5.11.9)]{DLMF}),
\be\label{eq:Stirling}
\Gamma\left(x+\i y\right)\sim \sqrt{2\pi}|y|^{x-\frac{1}{2}}\e^{-\frac{\pi}{2}|y|},\qquad x,y\in\mathbb R,\ y\to \pm\infty.
\ee

The {homogenenous} Log Gamma polymer occurs when $\alpha_k=\alpha>0$ for $k=1,\ldots, n$ and $a_m=0$ for $m=1,\ldots, N$, and is of particular interest. Observe that this case meets the assumptions also when $n=N$ and $\alpha<2$, since we can then take $\delta_2>\alpha/2$. The stronger decay needed in parts 3 and 4 of Theorem \ref{thm:Fredholm} is also valid for $\alpha<2$ if $n=N$ is sufficiently large.

\begin{remark}
The {analogue of the kernel $K_{N,t}$ in \cite{BCR}} was defined without squaring $2\pi\i$ in the denominator, but this is due to an extra factor $\frac{1}{2\pi\i}$ in the definition of the scalar product on $L^2({\Sigma_N})$ there.
\end{remark}

\paragraph{Associated biorthogonal measures.}
As anticipated in the introduction, Remark \ref{remark:large}, combining the Fredholm determinant identity \eqref{eq:LaplaceLogGamma} with Theorem \ref{thm:Fredholm}, part 4, we can relate the partition function of the Log Gamma polymer to an explicit biorthogonal measure. In the following corollary and below,  given $b \in \mathbb C$, $[b]_k$ will denote a $k$-dimensional vector with $b$ in each entry and $\vec e_m$ will denote the $N$-dimensional vector containing $1$ in position $m$ and $0$ elsewhere.{ We also recall the definition of Meijer $G$-function (see for instance \cite{DLMF})
$$
	\MeijerG*{m}{n}{p}{q}{a_1,\ldots,a_p}{b_1,\ldots,b_q}{z} := \frac{1}{2 \pi \i} \int_L \frac{\prod_{\ell = 1}^m \Gamma(b_\ell - v) \prod_{\ell = 1}^n \Gamma(1 - a_\ell + v)}{\prod_{\ell = m}^{q-1} \Gamma(1 - b_{\ell + 1} + v) \prod_{\ell = n}^{p-1} \Gamma(a_{\ell + 1} - v)}z^v \d v .
$$
The choice of the contour $L$ depends on the positions of the parameters $\{a_i, \, i = 1,\ldots,p\}$ and $\{ b_j,\, j = 1,\ldots,q\}$. For our specific cases below, we can choose it as the vertical line $\ell_N$.} \\

{\begin{corollary}\label{finalcor:LogGamma}
	Consider two sets of parameters $\vec\alpha = \left( \alpha_j\right)_{j = 1}^n$ and $\vec a = \left( a_k\right)_{k = 1}^N$ with $n \geq N$ and such that $\alpha_j - a_k > 0$ for all $j = 1, \ldots, n$ and $k = 1,\ldots,N$. Suppose moreover that there exist two positive constants $\delta_1,\delta_2$ such that $\delta_1 < \min\{ \delta_2, 1 - \delta_2\} < 1/2$ and such that equations \eqref{eq:assumptionLogG} are satisfied. In the case where $n=N$ we assume in addition that $\sum_{k=1}^N(\alpha_j+a_j-2\delta_2)<-1$. Let $W_N$ be as in \eqref{def:FLogGamma}, $\Sigma_N$ the circle of radius $\delta_1$ around the origin, and $\ell_N$ the vertical line $\mathrm{Re} z = \delta_2$.
	Then:
\begin{enumerate} 
	\item 
			\begin{equation}\label{eq:LaplaceLogGamma1}
				\mathbb E \left[ {\rm e}^{-\e^t Z_{n,N}^{{\rm Log}\Gamma}(\valpha,\va)}\right] =\mu_N[\sigma_t],\qquad \sigma_t(x)=\frac{1}{1+\e^{-x-t}},
			\end{equation}
			where the measure $\mu_N$ is given by equations \eqref{def:BiOE-F}--\eqref{def:L}.
	\item	
			\begin{equation}\label{eq:LaplaceLogGamma2}
				\frac{\d}{\d t}\log\mathbb E \left[ {\rm e}^{-\e^t Z_{n,N}^{{\rm Log}\Gamma}(\valpha,\va)}\right] =\int_{\mathbb R}\sigma_t(x)\kappa_{N,t}(x)\d x,
			\end{equation}
			where $\kappa_{N,t}$ is the one-point function \eqref{def:onepoint} of the deformed biorthogonal measure \eqref{def:BiOE-F-deformed}.
	\item In the case where the parameters $\vec a$ are all distinct, the biorthogonal measure is given explicitly, in terms of Meijer $G$-function, by 
		\be\label{eq:BiOLogGamma}
			\d\mu_{n,N}^{{\rm Log}\Gamma}(\vec x;\vec\alpha,\vec a)=\frac{1}{Z_N}\det\left(\e^{a_mx_k}\right)_{k,m=1}^N
			\det\left(\MeijerG*{n}{0}{0}{n + N}{-}{\valpha ; [1]_m + \va-\vec e_m}{\e^{-x_k}}\right)_{m,k=1}^N
			\prod_{k=1}^N\d  x_k.
		\ee
		In the confluent case $a_1 = \cdots = a_N = 0$, the biorthogonal measure is given by
		 \be\label{eq:BiOLogGammaconfluent}
			\d\mu_{n,N}^{{\rm Log}\Gamma}(\vec x;\vec\alpha,\vec 0)=\frac{1}{Z_N}\Delta(\vec x)\det\left(\MeijerG*{n}{0}{0}{n + N}{-}{\valpha ; [0]_{N-m+1}; [1]_{m-1}}{\e^{-x_k}}\right)_{m,k=1}^N
			\prod_{k=1}^N\d  x_k.
		\ee
\end{enumerate}
\end{corollary}}
{\proof The first part is a direct consequence of the Fredholm determinant identity \eqref{eq:LaplaceLogGamma} and Theorem \ref{thm:Fredholm}, while the second part is a corollary of Theorem \ref{thm: deformed}. As for the third part, one can compute explicitly the functions $\psi_m$ from \eqref{def:L1} as 
\begin{align*}\psi_m(y)&=\frac{1}{2\pi\i \left(W_{n,N}^{{\rm Log}\Gamma}\right)'(a_m;\vec \alpha,\vec a)}\int_{\ell_N}\frac{\prod_{j=1}^n\Gamma(\alpha_j-v)\ y^{-v}}{(v-a_m)\prod_{j=1}^N\Gamma(v - a_j)}\d v\\
&=\frac{1}{2\pi\i \left(W_{n,N}^{{\rm Log}\Gamma}\right)'(a_m;\vec \alpha,\vec a)}\int_{\ell_N}\frac{\prod_{j=1}^n\Gamma(\alpha_j-v)\ y^{-v}}{\Gamma(v-a_m+1)\prod_{j\neq m}\Gamma(v - a_j)}\d v,\end{align*}
and we recognize the right hand side as a Meijer $G$-function \cite{DLMF}:
\be\label{eq:psiMeijerG}\psi_m(y)=\frac{1}{\left(W_{n,N}^{{\rm Log}\Gamma}\right)'(a_m;\vec \alpha,\vec a)}\MeijerG*{n}{0}{0}{n + N}{-}{\valpha ; [1]_N + \va-\vec e_m}{1/y}.\ee
In the confluent case where $a_1=\cdots=a_N=0$, we have similarly that the functions $\phi_m$ in \eqref{def:phiF} are given by
\begin{align*}\phi_m(y)&=\frac{1}{2\pi\i}\int_{\ell_N}\frac{\prod_{j=1}^n\Gamma(\alpha_j-v)}{\Gamma(v)^N}\frac{y^{-v}}{v^{N-m+1}}\d v\\&=\frac{1}{2\pi\i}\int_{\ell_N}\frac{\prod_{j=1}^n\Gamma(\alpha_j-v)}{\Gamma(v)^{m-1}\Gamma(v+1)^{N-m+1}}{y^{-v}}\d v\\&=\MeijerG*{n}{0}{0}{n + N}{-}{\valpha ; [0]_{N-m+1}; [1]_{m-1}}{1/y}.\end{align*}
\endproof} 
{
\begin{remark}
It is easy to see, using the definition of Meijer $G$-functions, that when $W_N$ is given by \eqref{def:FLogGamma}, the associated functions $\{\psi_i\}_{i = 1}^N$ and $\{\phi_i\}_{i = 1}^N$, defined by \eqref{def:L1} and \eqref{def:phiF}, are real valued:
$$\psi_i({\rm e}^{\bar x}) = \overline{\psi_i({\rm e}^{x})}, \quad \phi_i({\rm e}^{\bar x}) = \overline{\phi_i({\rm e}^{x})},  \quad i = 1,\ldots, N.$$
Consequently, $d\mu_{n,N}^{{\rm Log}\Gamma}$ is a signed real-valued measure.
\end{remark}}

\paragraph{LUE with external source.}
{Let $B$ be a deterministic positive-definite Hermitian $N\times N$ matrix with  eigenvalues $b_1,\ldots, b_N\in[0,1)$, and consider 
the probability distribution
\[\frac{1}{Z_N'}(\det M)^\nu \e^{-{\rm Tr}((I-B) M)}\d M,\qquad \nu>0,\quad \d M=\prod_{j=1}^N\d M_{jj}\ \prod_{1\leq j<k\leq N}\d\Re M_{jk}\,\d \Im M_{jk},\]
on the space of $N\times N$ positive-definite Hermitian matrices. This model is known as the chiral Unitary Ensemble or LUE with external source and is connected to sample covariance matrices, see e.g.\ \cite{BBP, DesrosiersForrester, ElKaroui, Forrester}.
For $B=0$, this is the classical LUE or Wishart-Laguerre ensemble.} 

\medskip

{The eigenvalues of $M$ have the joint probability distribution
\be\label{eq:jpdfLUE+}\frac{1}{Z_N}\Delta(\vec x)\ \det\left(x_j^{\nu}\e^{-(1-b_k) x_j}\right)_{j,k=1}^N
\prod_{j=1}^N\d x_j,\qquad x_1,\ldots, x_N>0.\ee
This is a determinantal point process with kernel (see \cite[Formula (4.2)]{DesrosiersForrester} or \cite[Formula (5.158)]{Forrester})
\be\label{eq:LLUE+0}
L_{N}^{\rm LUE+}(x,x';\vec b,\nu)  := \frac{\e^{x'-x}}{(2 \pi \i)^2} \int_{\gamma_1} \d u \int_{\gamma_2} \d v \left(\frac{u}{v}\right)^{N+\nu}\prod_{j=1}^N\frac{v-1+b_j}{u-1+b_j} \frac{\e^{-ux'+vx}}{v-u},\qquad x,x'>0,
\ee
where $\gamma_1$ encircles $1-b_1,\ldots, 1-b_N$ anti-clockwise, $\gamma_2$ encircles $0$ anti-clockwise, and they don't intersect. 
Changing $u\mapsto 1-u$, $v\mapsto 1-v$, this is equal to 
\be\label{eq:LLUE+}
L_{N}^{\rm LUE+}(x,x';\vec b,\nu)  := \frac{1}{(2 \pi \i)^2} \int_{\Sigma_N} \d u \int_{\ell_N} \d v \left(\frac{u-1}{v-1}\right)^{N+\nu}\prod_{j=1}^N\frac{v-b_j}{u-b_j} \frac{\e^{-vx+ux'}}{v-u},
\ee
where $\Sigma_N$ is a positively oriented circle around $b_1,\ldots, b_N$ with radius $<1$, and $\ell_N$ a negatively oriented circle around $1$, which can moreover be deformed (after changing the orientation) to a vertical line lying between $\Sigma_N$ and $1$ if $\nu>0$.
This kernel is of the form \eqref{def:L} with 
\be\label{def:FLUE+}
W_N(z) \equiv W_{N}^{\rm LUE+}(z;\vec b,\nu):=\frac{\prod_{j=1}^N(z-b_j)}{(z-1)^{N+\nu}}.
\ee
Biorthogonal systems associated to this kernel can be constructed in terms of multiple Laguerre polynomials, see \cite{BleherKuijlaars}.
We will now explain that \eqref{eq:LLUE+}, with $\nu=n-N$ a non-negative integer, is the \emph{zero-temperature limit} of the Log Gamma polymer biorthogonal measure.}

\paragraph{Zero-temperature limit.}
{Let us consider the Log Gamma polymer with parameters
$$\alpha_1=\cdots=\alpha_n=T\to 0$$ 
and re-scaled $a$-parameters $a_j = Tb_j, \quad j = 1\ldots, N$ with $0\leq b_1,\ldots, b_N<1$. 
We then have
\begin{align*}W^{{\rm Log}\Gamma}_{n,N}(T\zeta; [T]_n, T\vec b)&=\frac{\Gamma(T(1-\zeta))^n}{\prod_{m=1}^N\Gamma(T(\zeta-b_m))}
\\
&=T^{N-n}
\frac{\prod_{m=1}^N(\zeta-b_m)}{(1-\zeta)^n}\frac{\Gamma(1+T(1-\zeta))^n}{\prod_{m=1}^N\Gamma(1+T(\zeta-b_m))}.
\end{align*}
and hence
\be\label{eq:scalinglimitLogGamma}\lim_{T\to 0}\frac{W^{{\rm Log}\Gamma}_{n,N}(Tv;[T]_n,T\vec b)}{W^{{\rm Log}\Gamma}_{n,N}(Tu;[T]_n,T\vec b)}=\frac{W_{N}^{\rm LUE+}(v;\vec b,n-N)}{W_{N}^{\rm LUE+}(u;\vec b,n-N)},\qquad\mbox{for $v\in\ell_N$, $u\in\Sigma_N$}.\ee}

A more careful analysis leads us to the following.

\begin{proposition}\label{prop:LogGamma}
{We have the point-wise limit
	\be\label{eq:zerotemplimit}\lim_{T\to 0}\frac{1}{T}L_{n,N}^{{\rm Log}\Gamma}\left(\frac{x}{T},\frac{x'}{T};[T]_n, T\vec b\right)=L_{N}^{\rm LUE+}\left(x,x';\vec b,\nu=n-N\right),\qquad x,x'\in\mathbb R,\ee
	and there exists $C_{n,N}(\vec b)$ such that for $\epsilon, T>0$ sufficiently small,
	\be\label{eq:dominatedcvgence}\left|\frac{1}{T}L_{n,N}^{{\rm Log}\Gamma}\left(\frac{x}{T},\frac{x'}{T};[T]_n, T\vec b\right)\right|\leq C_{n,N}(\vec b) h(x)\tilde h(x'), \ee
with	
	\be \label{def:h2}h(x)=\e^{-(b_{\max}+2\epsilon) x},\qquad \tilde h(x')=\begin{cases}\e^{(b_{\min}-\epsilon) x'},& x'<0\\
\e^{(b_{\max} +\epsilon)x'},&x'\geq 0\end{cases}.\ee
	}
\end{proposition}
\begin{proof}
{It is straightforward to verify that
\begin{align*}&\frac{1}{T}L_{n,N}^{{\rm Log}\Gamma}\left(\frac{x}{T},\frac{x'}{T};[T]_n, T\vec b\right)=\frac{1}{(2 \pi \i)^2} \int_{\Sigma_N} \d s \int_{\ell_N} \d w \frac{W_{n,N}^{{\rm Log}\Gamma}(Tw;[T]_n,T\vec b)}{W_{n,N}^{{\rm Log}\Gamma}(Ts;[T]_n,T\vec b)} \frac{\e^{-wx+sx'}}{w-s}\\
&\quad =\frac{1}{(2 \pi \i)^2} \int_{\Sigma_N} \d s \int_{\ell_N} \d w \left(\frac{s-1}{w-1}\right)^{n}\frac{\Gamma(1+T(1-w))^n}{\Gamma(1+T(1-s))^n}\prod_{j=1}^N\frac{w-b_j}{s-b_j}\frac{\Gamma(1+T(s-b_j))}{\Gamma(1+T(w-b_j))} \frac{\e^{-wx+sx'}}{w-s},\end{align*}
where $\Sigma_N$ is a contour surrounding $b_1,\ldots, b_N$ and no other zeros of $W_{n,N}^{{\rm Log}\Gamma}(Tu)$, and $\ell_N$ lies at the right of $\Sigma_N$ and at the left of $1$, and where we used $\Gamma(z+1)=z\Gamma(z)$ to pass to the second line.
Hence,
\begin{align*}
&\frac{1}{T}L_{n,N}^{{\rm Log}\Gamma}\left(\frac{x}{T},\frac{x'}{T};[T]_n, T\vec b\right)\\
&\quad =\frac{1}{(2 \pi \i)^2} \int_{\Sigma_N} \d s \int_{\ell_N} \d w \frac{W_{N}^{\rm LUE+}(w;\vec b,n-N)}{W_{N}^{\rm LUE+}(s;\vec b,n-N)} 
\left(\frac{\Gamma(1+T(1-w))^n}{\Gamma(1+T(1-s))^n}\prod_{j=1}^N\frac{\Gamma(1+T(s-b_j))}{\Gamma(1+T(w-b_j))}\right)
\frac{\e^{-wx+sx'}}{w-s}.
\end{align*}
The integrand converges point-wise to $0$ as $T\to 0$. Moreover, the part between brackets is uniformly bounded for $T$ sufficiently small, $s\in\Sigma_N$, and $w\in\ell_N$, provided that we take $\ell_N$ close to $1$. To see the uniformity in $w$, one has to use \eqref{eq:Stirling}. We can then use Lebesgue's dominated convergence theorem to conclude that we have the point-wise convergence \eqref{eq:zerotemplimit}.

\medskip

To obtain the bound \eqref{eq:dominatedcvgence}, we take $\ell_N=b_{\max}+2\epsilon+\i\mathbb R$, and $\Sigma_N$ within the strip $b_{\min}-\epsilon\leq \Re s\leq b_{\max}+\epsilon$, with $\epsilon>0$ sufficiently small so that $b_{\max}+2\epsilon<1$. For $T>0$ sufficiently small, it then follows that
\begin{align*}
\left|\frac{1}{T}L_{n,N}^{{\rm Log}\Gamma}\left(\frac{x}{T},\frac{x'}{T};[T]_n, T\vec b\right)\right|
&\leq 2\frac{h(x)\tilde h(x')}{(2 \pi)^2} \int_{\Sigma_N} |\d s| \int_{\ell_N} |\d w| \left|\frac{W_{N}^{\rm LUE+}(w;\vec b,n-N)}{W_{N}^{\rm LUE+}(s;\vec b,n-N)}\right| 
\frac{1}{|w-s|}\\
&\leq C_{n,N}(\vec b) h(x)\widetilde h(x').\end{align*}
The latter inequality holds only if $n>N$, since $\left|\frac{W_{N}^{\rm LUE+}(w;\vec b,n-N)}{W_{N}^{\rm LUE+}(s;\vec b,n-N)}\right| 
\frac{1}{|w-s|}$ is not integrable for $n=N$. One can however refine the above estimates by taking into account the decay
of
$\frac{\Gamma(1+T(1-w))^n}{\Gamma(1+T(1-s))^n}\prod_{j=1}^N\frac{\Gamma(1+T(s-b_j))}{\Gamma(1+T(w-b_j))}$ as $w\to \infty$ on $\ell_N$, such that the same inequality remains to hold.
}
\end{proof}
{\begin{corollary}\label{cor:LogGamma}
As $T\to 0$, we have the convergence 
	\be\label{eq:LUElimit}\mathbb E \left[ {\rm e}^{-\e^{t/T} Z_{n,N}^{{\rm Log}\Gamma}([T]_n,T\vec b)}\right]\longrightarrow \mathbb P_{\mathrm{LUE+}}\left(\max\{x_1,\ldots, x_N\}\leq -t ; \vec b, \nu = {n - N} \right),\ee
	where the expectation on the left is with respect to the distribution of the Log Gamma polymer partition function, and the probability on the right is with respect to the LUE with external source \eqref{eq:jpdfLUE+}.	
\end{corollary}
\begin{proof}
As a consequence of the previous result, we have point-wise convergence of the determinants
\[\det\left(\frac{\sigma_{t/T}(x_j/T)}{T}L_{n,N}^{{\rm Log}\Gamma}\left(\frac{x_j}{T},\frac{x_k}{T};[T]_n, T\vec b\right)\right)_{j,k=1}^N\to \det\left(1_{(-t,+\infty)}(x)L_{N}^{\rm LUE+}\left(x_j,x_k;\vec b,n-N\right)\right)_{j,k=1}^N\]
as $T\to 0$, where $\sigma_{t/T}(x/T)=\frac{1}{1+\e^{-\frac{x+t}{T}}}$. Moreover, we have the domination
\[\left|\det\left(\frac{\sigma_{t/T}(x_j/T)}{T}L_{n,N}^{{\rm Log}\Gamma}\left(\frac{x_j}{T},\frac{x_k}{T};[T]_n, T\vec b\right)\right)_{j,k=1}^N\right|\leq C_{n,N}^N\prod_{j=1}^Nh(x_j)\tilde h(x_j)\sigma_{t/T}(x_j/T)\ |\det O(1)|,\]
where $O(1)$ here denotes an $N\times N$ matrix with uniformly bounded entries.
Integrating over $\mathbb R^N$ yields
$\mu_{n,N}^{{\rm Log}\Gamma}[\sigma_{t/T}(./T)]\to \mu_{N}^{\rm LUE+}[1_{(t,+\infty)}]$, which implies the result by \eqref{eq:LaplaceLogGamma1}.
\end{proof}
}

\begin{remark}
The right hand side of \eqref{eq:LUElimit} is the largest eigenvalue distribution in the LUE with external source, and it is well-known \cite{Johansson2, BorodinPeche, DiekerWarren} that this distribution characterizes last passage percolation with exponential weights, which is the zero temperature limit of the Log Gamma polymer. The above results are thus not surprising, but reveal how the zero-temperature limit takes place on the level of the associated biorthogonal measures. Moreover, it now becomes clear that \eqref{eq:LaplaceLogGamma1} is the natural finite temperature generalization of the results from \cite{Johansson2, BorodinPeche, DiekerWarren} which express the Laplace transform of last passage percolation in terms of the largest eigenvalue of the LUE with external source: the biorthogonal measures \eqref{eq:BiOLogGamma} and \eqref{eq:BiOLogGammaconfluent} play a similar role in the Log Gamma polymer as the LUE with external source in last passage percolation. 
\end{remark}
\begin{remark}
Corollary \ref{cor:LogGamma} is consistent with known large $N$ asymptotic results from \cite{BCR, BCD}, stating that suitably rescaled Log Gamma polymer partition function has Tracy-Widom and Baik-Ben Arous-Péché fluctuations. These types of behavior are well-known to occur also in the large $N$ limits of the LUE and the LUE with external source, respectively.
\end{remark}

\section{O'Connell-Yor polymer partition function and GUE with external source}\label{section:OY}

\paragraph{The model.} The O'Connell-Yor polymer \cite{OConnelYor}, or semi-discrete directed polymer, is a semi-discrete version of the Log Gamma polymer. In order to define it, we fix a real positive number $\tau > 0$, a positive integer $N > 0$ and real numbers $a_1,\ldots,a_N$. 
For any $i = 1, \ldots, N$, $B_i$ will denote a one-dimensional Brownian motion with drift $a_i$. In other words, for any $s \geq 0$, $B_i(s) \sim a_i s + B(s)$, where $B(s)$ is a standard one-dimensional Brownian motion. Moreover, we assume that $B_1,\ldots, B_N$ are independent. 
A semi-discrete up-right path between $(0,1)$ and $(\tau,N)$ consists of horizontal segments that can jump vertically between certain points $(s,k)$ and $(s,k+1)$, see Figure \ref{figure: OY}. The energy of such a path $\phi$ is given by the sum of the Brownian increments between the jump points,
\[E(\phi)=B_1(s_1) + (B_2(s_2) - B_2(s_1)) + \ldots + (B_N(\tau) - B_N(s_{N-1})).\]
The O'Connell-Yor polymer partition function is then defined as 
\begin{equation}
	Z_N^{\rm OY}(\vec a, \tau) := \int_{0 < s_1 < \ldots < s_{N-1} < \tau} {\rm e}^{E(\phi)} \d s_1 \cdots \d s_{N-1}.
\end{equation}

\begin{figure}[t]
\begin{center}
    \setlength{\unitlength}{1truemm}
    \begin{picture}(100,50)(0,-10)
    \put(5,2){$(0,1)$} \put(85,51){$(\tau,N)$}
    \put(10,0){\thicklines\circle*{.8}}
    
    \put(90,50){\thicklines\circle*{.8}}

    \put(10,0){\line(1,0){22}}\put(22.7,0){\thicklines\vector(1,0){.0001}}
    \put(32,10){\line(1,0){12}}\put(39.7,10){\thicklines\vector(1,0){.0001}}
    \put(44,20){\line(1,0){16}}\put(53.7,20){\thicklines\vector(1,0){.0001}}
    \put(60,30){\line(1,0){10}}\put(66.7,30){\thicklines\vector(1,0){.0001}}
    \put(70,40){\line(1,0){8}}\put(75.7,40){\thicklines\vector(1,0){.0001}}
    \put(78,50){\line(1,0){12}}\put(85.7,50){\thicklines\vector(1,0){.0001}}
    
    \put(32,0){\line(0,1){10}}\put(32,6.7){\thicklines\vector(0,1){.0001}}
    \put(44,10){\line(0,1){10}}\put(44,16.7){\thicklines\vector(0,1){.0001}}
    \put(60,20){\line(0,1){10}}\put(60,26.7){\thicklines\vector(0,1){.0001}}
    \put(70,30){\line(0,1){10}}\put(70,36.7){\thicklines\vector(0,1){.0001}}
    \put(78,40){\line(0,1){10}}\put(78,46.7){\thicklines\vector(0,1){.0001}} 
    
    \put(26,-3){$(s_1,1)$}    \put(26,11){$(s_1,2)$}
        \put(38,7){$(s_2,2)$}    \put(38,21){$(s_2,3)$}
        \put(54,17){$(s_3,3)$}    \put(54,31){$(s_3,4)$}
        \put(64,27){$(s_4,4)$}    \put(64,41){$(s_4,5)$}
        \put(72,37){$(s_5,5)$}    \put(72,51){$(s_5,6)$}
    
    \put(1,-1){$a_1$}
    \put(1,9){$a_2$}
    \put(1,19){$a_3$}
    \put(1,29){$a_4$}
    \put(1,39){$a_5$}
    \put(1,49){$a_6$}

%

    
    \end{picture}
    \caption{A possible semi-discrete up-right path, with $N=6$, and jump points $s_1,\ldots, s_5$. The weight of a path with jump at $(s_k,k)$ in the O'Connell-Yor polymer depends on the value of $a_k$.}
    \label{figure: OY}
\end{center}
\end{figure}
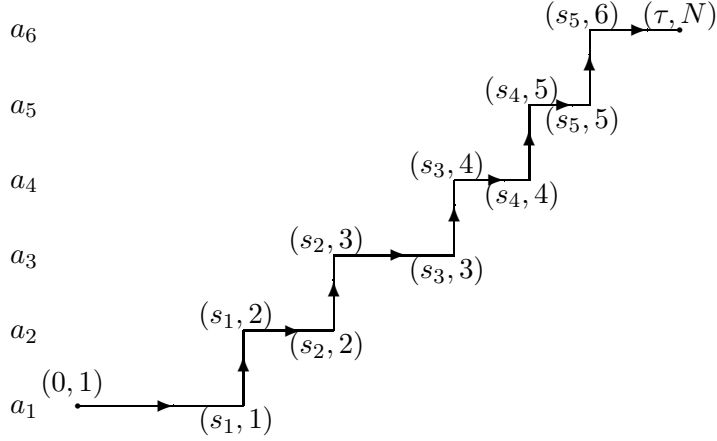

\paragraph{Fredholm determinant identity.}
{Suppose that $|a_i| < \delta_1<1/2$ for all $i=1,\ldots, N$, and take $\delta_2\in(0,1)$ such that $\delta_2>2\delta_1$. Then, Borodin and Corwin \cite[Theorem 5.2.11]{BC} proved the identity
\begin{equation}\label{eq:LaplaceOY}
		\mathbb E \left[ {\rm e}^{-\e^t Z_{N}^{\rm OY}(\va,\tau)}\right] = \det(1 + K_{N,t})_{L^2({\Sigma_N})},
\end{equation}
with
$K_{N,t}$ as in \eqref{def:kernel}, with ${\Sigma_N}$ the circle of radius $\delta_1$ around the origin, $\ell_N$ the vertical line $\Re z=\delta_2$, and $W_N$ given by
 \be\label{def:FOY}
W_N(z)=W_{N}^{\rm OY}(z;\vec a,\tau):=\frac{\e^{\frac{\tau z^2}{2}}}{\prod_{k=1}^N\Gamma(z-a_k)}.
\ee
It is easy to check that $W_N$ satisfies Assumption \ref{assumptions}, because of the decay of $\e^{\frac{\tau z^2}{2}}$ as $\Im z\to \pm\infty$.

\paragraph{Associated biorthogonal measures.}
{In the same fashion as for the Log Gamma polymer, we can combine the identity \eqref{eq:LaplaceOY} with part 4 of Theorem \ref{thm:Fredholm} to relate the partition function of the O'Connell-Yor polymer to an explicit biorthogonal measure. In the homogeneous case, the following corollary was already established by Imamura and Sasamoto in \cite[Theorem 2]{ImamuraSasamoto1}. 
\begin{corollary}\label{finalcor:OY}
	 Consider a set of real parameters $\vec a = (a_i)_{i = 1}^N$ with
$|a_i| < \delta_1<1/2$ for all $i=1,\ldots, N$, and take $\delta_2\in(0,1)$ such that $\delta_2>2\delta_1$.	
Let $W_N$ be as in \eqref{def:FOY}, ${\Sigma_N}$ the positively oriented circle of radius $\delta_1$ around $0$,  and $\ell_N$ the vertical line with real part $\delta_2$.
Then:
\begin{enumerate}
	\item 
			\begin{equation}\label{eq:LaplaceOY1}
				\mathbb E \left[ {\rm e}^{-\e^t Z_{N}^{{\rm OY}}(\va,\tau)}\right] =\mu_N[\sigma_t],\qquad \sigma_t(x)=\frac{1}{1+\e^{-x-t}},
			\end{equation}
			where the measure $\mu_N$ is given by equations \eqref{def:BiOE-F}--\eqref{def:L}.
	\item	
			\begin{equation}\label{eq:LaplaceOY2}
				\frac{\d}{\d t}\log\mathbb E \left[ {\rm e}^{-\e^t Z_{N}^{{\rm OY}}(\va,\tau)}\right] =\int_{\mathbb R}\sigma_t(x)\kappa_{N,t}(x)\d x,
			\end{equation}
			where $\kappa_{N,t}$ is the one-point function \eqref{def:onepoint} of the deformed biorthogonal measure \eqref{def:BiOE-F-deformed}.
	\item In the case where the parameters $\vec a$ are all distinct, the biorthogonal measure $\mu_N$ is given explicitly as
\[	\d\mu_{N}^{{\rm OY}}(\vec x;\vec a,\tau)=\frac{1}{Z_N}\det\left(\e^{a_mx_k}\right)_{k,m=1}^N
			\det\left(\frac{1}{2 \pi \i}\int_{\frac{1}2 + \i \mathbb R} \frac{\e^{\tau \frac{t^2}2} \e^{-tx_k}}{(t - a_m){\prod_{j = 1}^N \Gamma(t - a_j)}}\d t\right)_{m,k=1}^N
			\prod_{k=1}^N\d  x_k.\]
		In the confluent case $a_1=\cdots=a_N=a$, 
it is given by
\[\d\mu_{N}^{{\rm OY}}(\vec x;\vec a,\tau)=\frac{1}{Z_N}\det\left(x_k^{m-1}\right)_{k,m=1}^N
			\det\left(\frac{1}{2 \pi i}\int_{\frac{1}2 + \i \mathbb R} \frac{\e^{-x_k(t - a)} \e^{\tau \frac{t^2}2}}{\Gamma(t-a)^N(t - a)^{N - m + 1}} \d t\right)_{m,k=1}^N
			\prod_{k=1}^N\d  x_k.\]	

\end{enumerate}
	
\end{corollary}
{
\begin{remark}
Using the integral representations of $\{\psi_i\}_{i = 1}^N$ and $\{\phi_i\}_{i = 1}^N$ with $W_N$ given by \eqref{def:FOY}, it is easy to prove that
$$\psi_i({\rm e}^{\bar x}) = \overline{\psi_i({\rm e}^{x})}, \quad \phi_i({\rm e}^{\bar x}) = \overline{\phi_i({\rm e}^{x})},  \quad i = 1,\ldots, N.$$
Consequently, $d\mu_N^{{\rm OY}}$ is a signed real-valued measure.
\end{remark}}

\paragraph{GUE with external source.}%
{The $N\times N$ GUE with external source
can be defined as
\[ \frac{1}{Z_N}\e^{-{\rm Tr}\,\left(\frac{M^2}{2\tau} - AM \right)}\d M,\quad \d M=\prod_{j=1}^N\d M_{jj}\ \prod_{1\leq j<k\leq N}\d\Re M_{jk}\,\d \Im M_{jk}, \] 
where $A$ is a deterministic Hermitian $N\times N$ matrix, and $\tau>0$.
For $A=0$ and $\tau=1$, this is the classical GUE. In general, the random matrix  $M$ is the sum of a GUE matrix with the deterministic matrix $A$, see \cite{BrezinHikami, BrezinHikami2, Johansson, BleherKuijlaars}. The eigenvalue distribution is given by
\be\label{eq:jpdfGUE+}\frac{1}{Z_N}\Delta(\vec x)\ \det\left(\e^{-\frac{x_j^2}{2 \tau}+a_k x_j}\right)_{j,k=1}^N
\prod_{j=1}^N\d x_j,\qquad x_1,\ldots, x_N\in\mathbb R,\ee
if $a_1,\ldots, a_N$ are distinct, and is explicit in confluent cases as well.
For general $a_1,\ldots, a_N$, the eigenvalue distribution is a biorthogonal ensemble with kernel $L_N$ of the form \eqref{def:L}, with $W_N$ given by
\be
\label{def:FGUE+}W_N^{\rm GUE+}(z;\va,{\tau})=\prod_{m=1}^N(z-a_m)\ \e^{{\tau}z^2/2}.
\ee
see e.g. \cite{Forrester}. We will denote the associated kernel, defined by \eqref{def:L}, as $L_N^{\rm GUE+}$.
Biorthogonal systems associated to this measure can be constructed in terms of multiple Hermite polynomials, see  \cite{BleherKuijlaars}.}

\paragraph{$T \to 0$ limit.}
Consider the O'Connell-Yor polymer partition function with rescaled parameters tau parameter $\tau/T^2$ and with re-scaled $a_m$-parameters $T a_1,\ldots, Ta_N$.
Then, the kernel of the associated biorthogonal measure $L_N$ is given by
\[		L_N^{\rm OY}\left(x,x';T\vec a,\tau/T^2\right) = \frac{1}{(2 \pi \i)^2} \int_{\Sigma_N} \d u \int_{\ell_N} \d v \frac{W_N^{\rm OY}(v;T\vec a,\tau/T^2)}{W_N^{\rm OY}(u;T\va,\tau/T^2)} \frac{{\rm e}^{-vx + ux'}}{v - u},\]
or equivalently by
\[		\frac{1}T L_N^{\rm OY}(x/T,x'/ T; T\vec a,\tau/T^2) = \frac{1}{(2 \pi \i)^2} \int_{\widetilde{\Sigma}_N} \d s \int_{\widetilde\ell_N} \d t \frac{\widetilde W_N^{\rm OY}(t;\va,\tau)}{ \widetilde W_N^{\rm OY}(s;\va,\tau)} \frac{{\rm e}^{-tx + sx'}}{t-s},\]
with
\[\widetilde W_N^{\rm OY}(z;\vec a,\tau)=\frac{\e^{\tau z^2/2}}{\prod_{k=1}^N\Gamma(T(z-a_k))}=\frac{T^N \prod_{k=1}^N(z-a_k)}{\prod_{k=1}^N\Gamma((1+T(z-a_k)))}\e^{\tau z^2/2},\] 
and where
\[ \tilde \ell_N := T^{-1}\ell_N, \quad \widetilde \Sigma_N := T^{-1}\Sigma_N .\]
It is easy to see that 
\[\lim_{\tau\to\infty}\frac{\widetilde W_N^{\rm OY}(v;\va,\tau)}{\widetilde W_N^{\rm OY}(u;\va,\tau)}=\frac{W_N^{\rm GUE+}(v;\va,\tau)}{W_N^{\rm GUE+}(u;\va,\tau)},\]
where $W_N^{\rm GUE+}$ is defined in equation \eqref{def:FGUE+} above. We can now prove the following results similarly as for the  Log Gamma polymer, with the simplification that the integrands can be dominated more easily thanks to the Gaussian factors.
\begin{proposition}\label{prop:OY}
{We have the point-wise limit
	\be\label{eq:zerotemplimitOY}\lim_{T \to 0} T^{-1} L_{N}^{{\rm OY}}\left(x/T,x'/T; T\va,\tau/T^2\right)=L_{N}^{\rm GUE+}\left(x,x';\vec a, \tau \right),\qquad x,x'\in\mathbb R,\ee
	and there exists $C_{N}(\vec a)$ such that for $\epsilon>0$ sufficiently small and $T >0$ sufficiently small,
	\be\label{eq:dominatedcvgenceOY}\left| T^{-1} L_{N}^{{\rm OY}}\left(x/T,x'/T;T\va,\tau/T^2\right)\right|\leq C_{N}(\vec a) h(x)\tilde h(x'), \ee
	with $h,\tilde h$ as in \eqref{def:h2}.
	}
\end{proposition}

\begin{corollary}\label{cor:OY}	As $T \to 0$, we have the convergence
	\[\mathbb E \left[ {\rm e}^{-\e^{t/T} Z_{N}^{{\rm OY}}(T\va,\tau/T^2)}\right]\longrightarrow \mathbb P_{\mathrm{GUE+}}\left(\max\{x_1,\ldots, x_N\}\leq -t; \vec a, \tau \right),\]
	where the expectation on the left is with respect to the distribution of the O'Connell-Yor polymer partition function, and the probability on the right is with respect to the GUE with external source, equation \eqref{eq:jpdfGUE+}.
	\end{corollary}
\begin{remark}Like for the Log Gamma polymer, this result should not come as a surprise. The zero temperature version of the O'Connell-Yor polymer is a model of Brownian queues which can be characterized by the GUE \cite{Baryshnikov, GravnerTracyWidom}. The first identity in \eqref{eq:LaplaceOY1} is the finite temperature analogue of these characterizations, in which the role of the GUE (with external source) is played by the biorthogonal measure $\mu_N$.
\end{remark}

\section{Mixed polymer partition function}
\label{section:mixed}
\paragraph{The model.}
The mixed polymer, or O'Connell-Yor polymer with boundary sources, is a mixture of the Log Gamma and the O'Connell-Yor polymer. It was introduced by Borodin, Corwin, Ferrari, and Vet\H{o} in \cite{BCFV}.
Given parameters $n,N\in\mathbb N$, $\tau>0$, $\alpha_1,\ldots, \alpha_n$, and $a_1,\ldots, a_N$ with $\alpha_j-a_k>0$ for every $j=1,\ldots, n$ and $k=1,\ldots, N$, the mixed polymer partition function is defined as
\be\label{def:mixedpolymerpf}Z_{n,N}^{\rm Mixed}(\vec\alpha,\vec a,\tau)=\sum_{k=1}^NZ_{k,N}^{\rm Log\Gamma}(\vec\alpha,a_1,\ldots, a_k)Z_{N-k}^{\rm OY}(a_{k+1},\ldots, a_N,\tau).\ee
It can be interpreted as the partition function for up-right paths consisting of a discrete (Log Gamma polymer type) path from $(1,1)$ to a point $(N,k)$, concatenated with a semi-discrete (O'Connell-Yor polymer type) path from $(N,k)$ to $(N+\tau,n)$, where $N,n,\tau$ are given, and $k$ depends on the chosen path, see Figure \ref{figure: mixedpolymer}.

\paragraph{Fredholm determinant identity.}
As far as we know, a Laplace transform formula of the form \eqref{eq:LaplaceLogGamma} or \eqref{eq:LaplaceOY} has not appeared in the literature for the general mixed polymer partition function, although there are Fredholm determinant identities for the Laplace transform that look alike \cite{BCFV}, but with the kernel defined through integration over an unbounded contour instead of ${\Sigma_N}$ (recall Remark \ref{remark:large}). It is not clear whether such unbounded contour integral Fredholm determinant expressions can be related directly to our class of biorthogonal measures.
There is however another general Fredholm determinant identity obtained by Imamura and Sasamoto in \cite[Proposition 4.5]{ImamuraSasamoto2}:
	\begin{equation}\label{eq:Laplacemixed}
		\mathbb E \left[ {\rm e}^{-\e^t Z_{N,N}^{\rm Mixed}(\valpha,\va,\tau)}\right] = \det(1 -\sigma_t M)_{L^2(\mathbb R)},
	\end{equation}
where $\sigma_t(x)=\frac{1}{1+\e^{-x-t}}$, $M$ is given by
\be\label{def:Lmixed}
M(x,x')=\sum_{k=1}^N \Phi_k(x)\Psi_k(x'),
\ee
and
\begin{align}
&\Phi_k(x)=\frac{1}{2\pi\i}\int_{{\Sigma_N}}\d u\frac{\e^{ux-\tau u^2/2}}{u-a_k}\prod_{j=1}^{k-1}\frac{u-\alpha_j}{u-a_j}\prod_{j=1}^N\frac{\Gamma(1+u-a_j)}{\Gamma(1+\alpha_j-u)}, \label{eq:defPhi}\\
&\Psi_k(x)=\frac{\alpha_k-a_k}{2\pi}\int_{\mathbb R}\d w\frac{\e^{-\i w x-\tau w^2/2}}{\alpha_k-\i w}\prod_{j=1}^{k-1}\frac{\i w-a_j}{\i w-\alpha_j}\prod_{j=1}^N\frac{\Gamma(1+\alpha_j-\i w)}{\Gamma(1+\i w-a_j)}. \label{eq:defPsi}
\end{align}
For the above identity to hold, ${\Sigma_N}$ has to enclose all the points $a_1,\ldots, a_N$ but none of the poles of $\prod_{m=1}^N\Gamma(1+u-a_m)$. {Note that this is only possible when there are no $a_k$-values which differ by an integer. Below, we assume that this is the case. We mention however that our final results Corollary \ref{finalcor:mixed}, Proposition \ref{prop:mixed}, and Corollary \ref{cor:mixed} hold without this assumption. This can be seen by a continuity argument.}

\begin{figure}[t]
\begin{center}
    \setlength{\unitlength}{1truemm}
    \begin{picture}(100,50)(0,-10)
    \put(-5,2){$(1,1)$} \put(85,52){$(\tau+N,N)$}
  \put(0,50){\thicklines\circle*{.8}}    
      \put(0,40){\thicklines\circle*{.8}}
        \put(0,30){\thicklines\circle*{.8}}
          \put(0,20){\thicklines\circle*{.8}}
            \put(0,10){\thicklines\circle*{.8}}
              \put(0,0){\thicklines\circle*{.8}}
    
    \put(10,0){\thicklines\circle*{.8}}
\put(20,0){\thicklines\circle*{.8}}
\put(30,0){\thicklines\circle*{.8}}
\put(40,0){\thicklines\circle*{.8}}
\put(50,0){\thicklines\circle*{.8}}
    
  \put(10,10){\thicklines\circle*{.8}}
\put(20,10){\thicklines\circle*{.8}}
\put(30,10){\thicklines\circle*{.8}}
\put(40,10){\thicklines\circle*{.8}}
\put(50,10){\thicklines\circle*{.8}}

    \put(10,20){\thicklines\circle*{.8}}
\put(20,20){\thicklines\circle*{.8}}
\put(30,20){\thicklines\circle*{.8}}
\put(40,20){\thicklines\circle*{.8}}
\put(50,20){\thicklines\circle*{.8}}
    
    \put(10,30){\thicklines\circle*{.8}}
\put(20,30){\thicklines\circle*{.8}}
\put(30,30){\thicklines\circle*{.8}}
\put(40,30){\thicklines\circle*{.8}}
\put(50,30){\thicklines\circle*{.8}}
    
\put(10,40){\thicklines\circle*{.8}}
\put(20,40){\thicklines\circle*{.8}}
\put(30,40){\thicklines\circle*{.8}}
\put(40,40){\thicklines\circle*{.8}}
\put(50,40){\thicklines\circle*{.8}}
    
    \put(10,50){\thicklines\circle*{.8}}
\put(20,50){\thicklines\circle*{.8}}
\put(30,50){\thicklines\circle*{.8}}
\put(40,50){\thicklines\circle*{.8}}
\put(50,50){\thicklines\circle*{.8}}

    \put(90,50){\thicklines\circle*{.8}}  
    
    \put(0,0){\line(1,0){10}}\put(6.7,0){\thicklines\vector(1,0){.0001}}    
    \put(10,0){\line(1,0){10}}\put(16.7,0){\thicklines\vector(1,0){.0001}}
    \put(20,10){\line(1,0){10}}\put(26.7,10){\thicklines\vector(1,0){.0001}}
    \put(30,10){\line(1,0){10}}\put(36.7,10){\thicklines\vector(1,0){.0001}}
    \put(40,20){\line(1,0){10}}\put(46.7,20){\thicklines\vector(1,0){.0001}}
    \put(50,20){\line(1,0){6}}\put(54.7,20){\thicklines\vector(1,0){.0001}}
    \put(56,30){\line(1,0){14}}\put(64.7,30){\thicklines\vector(1,0){.0001}}
    \put(70,40){\line(1,0){8}}  \put(75.7,40){\thicklines\vector(1,0){.0001}}  
    \put(78,50){\line(1,0){12}} \put(85.7,50){\thicklines\vector(1,0){.0001}}
    
    \put(20,0){\line(0,1){10}}\put(20,6.7){\thicklines\vector(0,1){.0001}}
    \put(40,10){\line(0,1){10}}\put(40,16.7){\thicklines\vector(0,1){.0001}}
    \put(56,20){\line(0,1){10}}\put(56,26.7){\thicklines\vector(0,1){.0001}}
    \put(70,30){\line(0,1){10}}\put(70,36.7){\thicklines\vector(0,1){.0001}}
    \put(78,40){\line(0,1){10}}\put(78,46.7){\thicklines\vector(0,1){.0001}} 
    
\put(51,17){$(s_3,3)$}    \put(51,31){$(s_3,4)$}
\put(65,27){$(s_4,4)$}    \put(65,41){$(s_4,5)$}
\put(73,37){$(s_5,5)$}    \put(73,51){$(s_5,6)$}    
        \put(-2,-7){$\alpha_1$}
    \put(8,-7){$\alpha_2$}
    \put(18,-7){$\alpha_3$}
    \put(28,-7){$\alpha_4$}
    \put(38,-7){$\alpha_5$}
    \put(48,-7){$\alpha_6$}
    
    \put(-9,-1){$a_1$}
    \put(-9,9){$a_2$}
    \put(-9,19){$a_3$}
    \put(-9,29){$a_4$}
    \put(-9,39){$a_5$}
    \put(-9,49){$a_6$}

%

    
    \end{picture}
    \caption{An up-right path in the mixed polymer model, corresponding to $n=N=6$, $k=3$. The values $\alpha_j, a_k$ determine the weight of the point $(j,k)$ in the discrete part of the lattice, while the values $a_k$ determine the weight of jump points $(s_k,k)$ in the lattice.}
    \label{figure: mixedpolymer}
\end{center}
\end{figure}
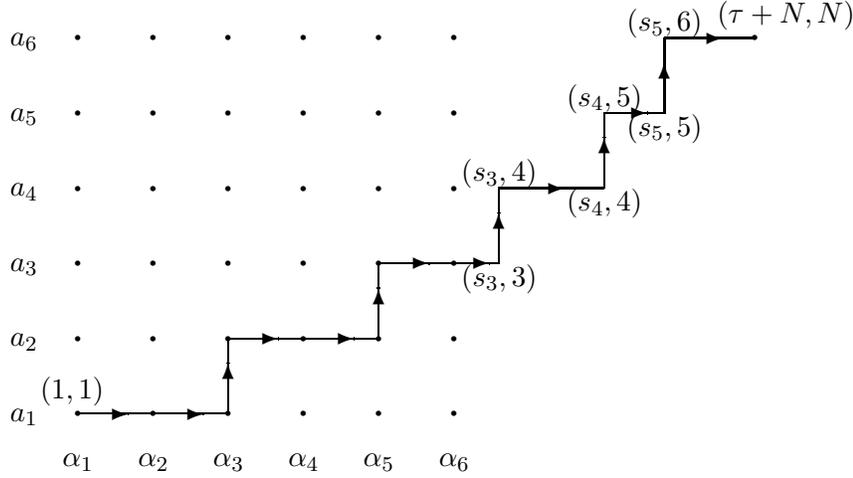

\paragraph{Associated biorthogonal measures.}
{We will show, starting from equation \eqref{eq:Laplacemixed}, that the mixed polymer partition function is also related to an explicit biorthogonal measure of the form \eqref{BiOE-distinct}. We start with the following 
\begin{lemma}
	Let $N \in \mathbb N, \{ a_1,\ldots,a_N \} $ and $ \{\alpha_1,\ldots,\alpha_N \}$ two sets of real parameters such that $\alpha_j - a_k > 0$ for all $j,k = 1, \cdots, N$. Then the two families of functions $\{ \Phi_k, \; k = 1, \ldots, N\}$ and $\{ \Psi_k, \; k = 1, \ldots N \}$, defined in \eqref{eq:defPhi}--\eqref{eq:defPsi}, are biorthogonal:
\[\int_{\mathbb R}\Phi_k(x)\Psi_m(x)\d x=\delta_{km},\qquad k,m=1,\ldots, N.\]
\end{lemma}
\begin{proof}
We first observe that $\Psi_m(x)$ is a Fourier transform:
\[\Psi_m(x)=\frac{\alpha_m-a_m}{2\pi}\mathcal F[g_m]\left(\frac{x}{2\pi}\right),\quad g_m(w)=\frac{\e^{-\tau w^2/2}}{\alpha_m-\i w}\prod_{j=1}^{m-1}\frac{\i w-a_j}{\i w-\alpha_j}\prod_{j=1}^N\frac{\Gamma(1+\alpha_j-\i w)}{\Gamma(1+\i w-a_j)},\]
where
\[\mathcal F[g](t):=\int_{\mathbb R}g(w)\e^{-2\pi \i w t}\d t,\quad \mathcal F^{-1}[G](w)=\int_{\mathbb R}G(t)\e^{2\pi \i w t}\d t.\]
It follows that
\begin{align*}\int_{\mathbb R}\Phi_k(x)\Psi_m(x)\d x&=-\frac{\alpha_m-a_m}{4\pi^2\i}\int_{{\Sigma_N}}\d u \frac{1}{(u-a_k)(u-\alpha_k)g_k(-\i u)}\int_{\mathbb R}\d x \e^{ux} \mathcal F[g_m]\left(\frac{x}{2\pi}\right)\d x\\
&=-\frac{\alpha_m-a_m}{2\pi\i}\int_{{\Sigma_N}}\d u \frac{1}{(u-a_k)(u-\alpha_k)g_k(-\i u)}\int_{\mathbb R}\d s \e^{2\pi s u} \mathcal F[g_m]\left(s\right)\d s\\
&=-\frac{\alpha_m-a_m}{2\pi\i}\int_{{\Sigma_N}}\d u \frac{1}{(u-a_k)(u-\alpha_k)g_k(-\i u)}\left(\mathcal F^{-1}\circ \mathcal F\right){[g_m]}\left(-\i u\right)\\&=-\frac{\alpha_m-a_m}{2\pi\i}\int_{{\Sigma_N}}\d u \frac{1}{(u-a_k)(u-\alpha_k)}\frac{g_m\left(-\i u\right)}{g_k(-\i u)}.\end{align*}
If $m=k$, this is equal to $-\frac{\alpha_m-a_m}{2\pi\i}\int_{{\Sigma_N}}\d u \frac{1}{(u-a_k)(u-\alpha_k)}=1$. If $m>k$, the integrand is analytic such that the integrand is $0$. If $m<k$, the integrand is a rational function without zeros in the exterior of ${\Sigma_N}$ and decaying like $O(1/u^2)$ as $u\to\infty$, such that we can deform ${\Sigma_N}$ to a large circle to see that the integral is $0$.
\end{proof}
Using the techniques from Section \ref{subsec:finitedet}, we can then conclude that the kernel $M$ defined in \eqref{def:Lmixed} is a kernel
satisfying the reproducing property, which defines
 the $N$-point biorthogonal measure
 \be\label{def:BiOE-M}
\d\widetilde\mu_N(x_1,\ldots, x_N):=\frac{1}{N!}
\det\left(M(x_m,x_k)\right)_{m,k=1}^N
\prod_{k=1}^N\d  x_k,
\ee
 which can be equivalently rewritten as 
 \be\label{eq:tildemu}
\d\widetilde\mu_N(\vec x)=\frac{1}{N!}
\det\left(\Phi_j(x_m)\right)_{m,j=1}^N
\det\left(\Psi_j(x_k)\right)_{j,k=1}^N
\prod_{k=1}^N\d  x_k.
\ee

We will now show that the above biorthogonal measure is of the form \eqref{def:BiOE-F}.

\begin{proposition}
Let $N \in \mathbb N, \{ a_1,\ldots,a_N \} $ and $ \{\alpha_1,\ldots,\alpha_N \}$ two sets of real parameters such that $\alpha_j - a_k > 0$ for all $j,k = 1, \cdots, N$.
The measure \eqref{eq:tildemu}
is equal to \eqref{def:BiOE-F} and \eqref{eq:BiOL}, with ${\Sigma_N}$ in the left half plane and enclosing $a_1,\ldots, a_N$ but none of the poles of $\prod_{j=1}^N\Gamma(1+v-a_j)$, $\ell_N$ a vertical line at the right of $\Sigma_N$, and $W_N$ given by
\be\label{def:Fmixedpolymer}
W_N(z)\equiv W_{N}^{\rm Mixed}(z;\vec\alpha,\vec a,\tau) :=\e^{\tau z^2/2}\frac{\prod_{j=1}^N\Gamma(\alpha_j-z)}{\prod_{k=1}^N\Gamma(z-a_k)}.
\ee
\end{proposition}
\begin{proof}
To show that $\widetilde\mu_N$ is equal to $\mu_N$ with $W_N$ given by \eqref{def:Fmixedpolymer}, it is sufficient to prove that these measures are equal in case $a_1,\ldots, a_N$ are distinct, by continuity.
In that case, we can easily evaluate the functions $\Phi_1, \ldots, \Phi_N$ using the residue theorem, which implies that they have the same linear span as $\e^{a_1 x},\ldots, \e^{a_N x}$. Similarly, we show that $\Psi_1, \ldots, \Psi_N$ have the same linear span as $\psi_1(\e^x),\ldots, \psi_N(\e^x)$ defined in \eqref{def:L1}: we have
{\begin{align*}\Psi_k(x)&=\frac{\alpha_k-a_k}{2\pi\i}\int_{\mathbb R}\d w \e^{-\i w x-\tau w^2/2}\frac{\prod_{j=1}^{k-1}(\i w-a_j)}{\prod_{j=1}^k{(\i w-\alpha_j)}}\prod_{j=1}^N\frac{\Gamma(1+\alpha_j-\i w)}{\Gamma(1+\i w-a_j)}\\
&=\pm \frac{\alpha_k-a_k}{2\pi\i}\int_{\mathbb R}\d w \e^{-\i w x-\tau w^2/2}\frac{\prod_{j=1}^{k-1}(\i w-a_j)\ \prod_{j=k+1}^N(\i w-\alpha_j)}{\prod_{j=1}^N{(\i w-\alpha_j)}\ }\prod_{j=1}^N\frac{\Gamma(\alpha_j-\i w)}{\Gamma(\i w-a_j)}\\
&=\pm \frac{\alpha_k-a_k}{2\pi}\int_{\ell_N}\d u \e^{-u x+\tau u^2/2}\frac{\prod_{j=1}^{k-1}(u-a_j)\prod_{j=k+1}^N{(u-\alpha_j)}}{\prod_{j=1}^N{(u-a_j)}}\prod_{j=1}^N\frac{\Gamma(\alpha_j-u)}{\Gamma(u-a_j)}\\
&=\pm \frac{\alpha_k-a_k}{2\pi}\int_{\ell_N}\d u \e^{-u x}\frac{\prod_{j=1}^{k-1}(u-a_j)\prod_{j=k+1}^N{(u-\alpha_j)}}{\prod_{j=1}^N{(u-a_j)}}W_N(u).\end{align*}}
On the other hand, we already know that 
$\psi_1(\e^x),\ldots, \psi_N(\e^x)$ span every function of the form
\[\int_{\ell_N}\d w \e^{-u x}\frac{P_{N-1}(u)}{\prod_{j=1}^N{(u-a_j)}}W_N(u),\]
with $P_{N-1}$ a polynomial of degree $\leq N-1$. We conclude that {$\Psi_1(x), \ldots, \Psi_N(x)$} and $\psi_1(\e^x),\ldots, \psi_N(\e^x)$ have the same linear span, such that
\eqref{eq:tildemu} can be rewritten as
\[\d\widetilde\mu_N(\vec x):=\frac{1}{N!}
\det\left(\e^{a_j x_m}\right)_{m,j=1}^N
\det\left(\psi_j(\e^{x_k})\right)_{j,k=1}^N
\prod_{k=1}^N\d  x_k,\]
which is equal to $\d\mu_N(\vec x)$ by \eqref{BiOE-distinct}. The result if proved.
\end{proof}
Combining the above results with Theorem \ref{thm: deformed}, we obtain the following characterization of the mixed polymer partition function in terms of the biorthogonal measures $\mu_N$ and $\nu_{N,t}$.

\begin{corollary}\label{finalcor:mixed}
	Let $N \in \mathbb N$, and let $\{ a_1,\ldots,a_N \} $, $ \{\alpha_1,\ldots,\alpha_N \}$ two sets of real parameters such that $\alpha_j - a_k > 0$ for all $j,k = 1, \cdots, N$.
Then:
\begin{enumerate}
	\item 
			\begin{equation}\label{eq:Laplacemixed2}
				\mathbb E \left[ {\rm e}^{-\e^t Z_{N,N}^{{\rm Mixed}}(\valpha,\va,\tau)}\right] =\mu_N[\sigma_t],\qquad \sigma_t(x)=\frac{1}{1+\e^{-x-t}},
			\end{equation}
			where the measure $\mu_N$ is given by equations \eqref{def:BiOE-F}, \eqref{def:L} with $W_N$ defined by \eqref{def:Fmixedpolymer}, with $\Sigma_N$ a loop in the left half plane enclosing $a_1,\ldots, a_N$ but none of the poles of $\prod_{j=1}^N\Gamma(1+v-a_j)$, and $\ell_N$ a vertical line at the right of $\Sigma_N$.
	\item	
			\begin{equation}\label{eq:logdermixed}
				\frac{\d}{\d t}\log\mathbb E \left[ {\rm e}^{-\e^t Z_{N,N}^{{\rm Mixed}}(\valpha,\va,\tau)}\right] =\int_{\mathbb R}\sigma_t(x)\kappa_{N,t}(x)\d x,
			\end{equation}
			where $\kappa_{N,t}$ is the one-point function \eqref{def:onepoint} of the deformed biorthogonal measure \eqref{def:BiOE-F-deformed}.
	\item In the case where the parameters $\vec a$ are all distinct,
the biorthogonal measure $\mu_N$ is given explicitly as
\[	\d\mu_{N}^{{\rm Mixed}}(\vec x;\vec\alpha,\vec a,\tau)=\frac{1}{Z_N}\det\left(\e^{a_mx_k}\right)_{k,m=1}^N
			\det\left(\frac{1}{2 \pi \i}\int_{\frac{1}2 + \i \mathbb R} \frac{{\rm e}^{\tau \frac{t^2}2} \prod_{j = 1}^n \Gamma(\alpha_j - t)}{(t - a_m) \prod_{k = 1}^N \Gamma(t - a_k)} {\e}^{-tx_k}\d t\right)_{m,k=1}^N
			\prod_{k=1}^N\d  x_k.\]
		In the confluent case $a_1=\cdots=a_N=a$, 
it is given by
\[\d\mu_{N}^{{\rm Mixed}}(\vec x;\vec\alpha,\vec a,\tau)=\frac{1}{Z_N}\det\left(x_k^{m-1}\right)_{k,m=1}^N
			\det\left(\frac{1}{2 \pi i}\int_{\frac{1}2 + \i \mathbb R} \frac{{\rm e}^{\tau \frac{t^2}2}{ \prod_{j = 1}^n\Gamma(\alpha_j - t)} {\e}^{-(t - a)x_k}}{(t - a)^{N - m +1} \Gamma(t - a)^N} \mathrm{d}t\right)_{m,k=1}^N
			\prod_{k=1}^N\d  x_k.\]			
\end{enumerate}
	
\end{corollary}
{
\begin{remark}
Using the integral representations of $\{\psi_i\}_{i = 1}^N$ and $\{\phi_i\}_{i = 1}^N$ with $W_N$ given by \eqref{def:Fmixedpolymer}, it is easy to prove that
$$\psi_i({\rm e}^{\bar x}) = \overline{\psi_i({\rm e}^{x})}, \quad \phi_i({\rm e}^{\bar x}) = \overline{\phi_i({\rm e}^{x})},  \quad i = 1,\ldots, N.$$
Consequently, $d\mu_N^{{\rm Mixed}}$ is a signed real-valued measure.
\end{remark}}

\begin{remark}
Corollary \ref{finalcor:mixed} contains both Corollary \ref{finalcor:LogGamma} with $n=N$  and Corollary \ref{finalcor:OY} as degenerate cases: the Log Gamma polymer corresponds to $\tau=0$, while the O'Connell-Yor polymer corresponds to $\alpha_1=\ldots=\alpha_N\to\infty$.
As already mentioned, the above identities are finite temperature generalizations of classical characterizations of last passage percolation and Brownian queues models in terms of the GUE and LUE with external sources. The biorthogonal measures $\mu_N$ replace the GUE and LUE with external source in the corresponding finite temperature models. The identity \eqref{eq:Laplacemixed2} can also be seen as a finite $N$ version of the characterization from \cite{ACQ} of the narrow wedge solution of the KPZ equation (or continuum directed polymer) in terms of the Airy point process. In this sense, the biorthogonal measures $\mu_N$ play the same role in (semi-)discrete directed polymer models as the Airy point process in the continuum directed polymer model.
Other instances of these relations between stochastic models and point processes are available in the literature. In \cite{Borodin6v}, the stochastic higher spin six vertex model is related to the MacDonald measure, which specializes to the Schur measure in the case of the homogenenous six vertex model. According to the author, the determinantal structure of the O'Connel-Yor polymer \cite{ImamuraSasamoto1} could be a degeneration of the one for the stochastic higher spin six vertex model. Other types of degenerations lead to characterize the Asymmetric Exclusion Process (ASEP) in terms of the discrete Laguerre ensemble, and of the discrete Hermite ensemble in a large time regime \cite{BorodinOlshanskiASEP}. Back to the KPZ equation, the characterization of the narrow wedge solution of the KPZ equation (or continuum directed polymer) in terms of the GUE Airy point process extends to half space case initial data, upon replacing the GUE Airy point process with the GOE Airy point process, which is not determinantal but Pfaffian \cite{BBCW,Par17}. In the discrete case, thanks to combinatorial identities between $q$-Whittaker measures and periodic and free boundary Schur measures \cite{ImamuraMucciconiSasamoto1}, Pfaffian Fredholm determinants characterize the $q$-PushTASEP and the Log Gamma polymer in half space \cite{ImamuraMucciconiSasamoto2}.
\end{remark}

\paragraph{Sum of LUE with external source with GUE.}
Consider an $N\times N$ matrix $M$ from the LUE with external source, with distribution \eqref{eq:jpdfLUE+} and $\nu = 0$, and consider an $N\times N$ GUE matrix $H$, with distribution \eqref{eq:jpdfGUE+} with $A=0$. Set $Q=M+\sqrt{\tau N}H$.
The eigenvalue distribution of $Q$ can be computed using \cite[Theorem 2.3 (b)]{ClaeysKuijlaarsWang}, or more directly using the re-scaled variant of this identity given in \cite[Equation (1.6)]{ClaeysDoeraene}. This result gives an integral expression for the correlation kernel of the eigenvalues of $Q$ in terms of the eigenvalue correlation of $M$, valid for any random matrix $M$ whose eigenvalue distribution is a polynomial ensemble, and thus in particular for our LUE with external source random matrix $M$.
It states that an eigenvalue correlation kernel $K_N(x,x')$ of $Q$ is given by
\be\label{eq:transformationKN}
K_N(x,x')=\frac{1}{2\pi\i\tau}\int_{\i\mathbb R}\d s \int_{\mathbb R}\d t \; k_N(s,t)\e^{\frac{1}{2\tau}\left((x'-s)^2-(x-t)^2\right)}, 
\ee
where $k_N(s,t)$ is a correlation kernel of the eigenvalue distribution of $M$ which is polynomial in $s$: more precisely, it is of the form
\be\label{knPEform}k_N(s,t)=\sum_{k=1}^Np_k(s)q_k(t),\ee
where $p_1,\ldots, p_N$ are polynomials of degree $N-1$, and we have the biorthogonality relations
\[\int_{\mathbb R}p_j(s)q_k(s)\d s=\delta_{jk},\qquad j,k=1,\ldots, N.\]
In our situation, we can take 
{\be
k_{N}(s,t) = \frac{1}{(2 \pi \i)^2} \int_{\Sigma} \d u \int_{\ell} \d v \left(\frac{v}{u}\right)^{N}\prod_{j=1}^N\frac{u-1+b_j}{v-1+b_j} \frac{\e^{us-vt}}{v-u},\qquad s,t\in\mathbb R
\ee
where $\Sigma$ encircles $0$, and $\ell=c+\i\mathbb R$ is a vertical line at the right of $\Sigma$ and at the left of $1-b_1,\ldots, 1-b_N$. The $v$-integral is defined as a Cauchy principal value integral, $\int_\ell=\lim_{R\to\infty}\int_{c-\i R}^{c+\i R}\d v$.
{Indeed,} for $t<0$, we can compute this principal value integral by closing the integration contour $[c-\i R, c+\i R]$ with a large semi-circle to the left, and it follows then from a residue computation that $k_N(s,t)=0$. Similarly, for $t>0$, we can close the integration contour with a large semi-circle on the right, and then we obtain that 
\[k_N(s,t)=\e^{s-t}L_N(s,t),\]
with $L_N=L_N^{\rm LUE+}$ given by
\eqref{eq:LLUE+0} with $\nu=0$.
It is also straightforward to verify using the residue theorem that $k_N(s,t)$ is polynomial in $s$, and more precisely of the required form \eqref{knPEform}. {As for the biorthogonality relations, we refer to \cite{BleherKuijlaars}.}
We then apply \eqref{eq:transformationKN} and find
that $K_N(x,x')$ is given by
\begin{align*}
K_N(x,x')&=\frac{1}{2\pi\i\tau}\int_{\i\mathbb R}\d s \int_{\mathbb R}\d t \; k_N(s,t)\e^{\frac{1}{2\tau}\left((x'-s)^2-(x-t)^2\right)}
\\&= \frac{1}{(2\pi\i)^3\tau}\int_{\i\mathbb R}\d s \int_{\mathbb R}\d t \int_{\Sigma}\d u\int_{\ell}\d v\left(\frac{v}{u}\right)^{N}\prod_{j=1}^N\frac{u-1+b_j}{v-1+b_j}\frac{\e^{us-vt}}{v-u}\e^{\frac{1}{2\tau}\left((x'-s)^2-(x-t)^2\right)}.
\end{align*}
Changing the order of integration and writing $\tilde W_N(z)=\frac{z^N}{\prod_{j=1}^N(z-1+b_j)}$, we obtain after a straightforward computation that
\begin{align*}
K_N(x,x')&= \frac{1}{(2\pi\i)^3\tau}\int_{\Sigma}\d u\int_{\ell}\d v\frac{\tilde W_N(v)}{\tilde W_N(u)}\frac{\e^{\frac{\tau}{2} v^2-\frac{\tau}{2} u^2+ux'-vx}}{v-u}\\
&\qquad\qquad\qquad \times\
\left(\int_{\i\mathbb R}\d s \; 
\e^{\frac{1}{2\tau}\left(x'-s-\tau u\right)^2}\right)
\left(\int_{\mathbb R}\d t \;
\e^{-\frac{1}{2\tau}\left(x-t-\tau v\right)^2}\right)\\
&= \frac{1}{(2\pi\i)^2}\int_{\Sigma}\d u\int_{\ell}\d v\frac{\tilde W_N(v)\e^{\frac{\tau}{2}v^2}}{\tilde W_N(u)\e^{\frac{\tau}{2}u^2}}\frac{\e^{x'u-xv}}{v-u}.
\end{align*}
We conclude that $K_N(x,x')$ is again a kernel of the form \eqref{def:L}, but now with transformed function $W$ given by
\begin{equation}\label{eq:defGLUE+}
	W_N^{\rm GLUE+}(z;\vec b,\tau) := \frac{z^N{\rm e}^{\tau z^2/2}}{\prod_{k = 1}^N (z-1+b_k)}.
\end{equation}
}

We will now point out that this biorthogonal measure corresponds to  the {small temperature} limit of the biorthogonal measure associated to the mixed polymer.

\paragraph{Small temperature limit for the mixed polymer.}
{Consider the mixed polymer partition function with re-scaled $\tau=\tau/T^2$, $a_1=\cdots=a_N=0$, and with re-scaled $\alpha$-parameters $\alpha_1=T-T b_1,\ldots, \alpha_N=1 -Tb_N$, in the limit where $T\to 0$.
Then, the kernel of the associated biorthogonal measure $L_N$ is given by
\[		{L_N^{\rm Mixed}\left(x,x';\left[ T\right]_N - T \vec b,[0]_N,\tau/T^2\right)}  = \frac{1}{(2 \pi \i)^2} \int_{\Sigma_N} \d u \int_{\ell_N} \d v \frac{W_N^{\rm Mixed}(v;[T]_N-T\vec b,[0]_N,\tau/T^2)}{W_N^{\rm Mixed}(u;[T]_N-T\vec b,[0]_N,\tau/T^2)} \frac{{\rm e}^{-vx + ux'}}{v - u},\]

or equivalently by
\[		\frac{1}{T} L_N^{\rm Mixed}\left(\frac{x}{T},\frac{x'}{T};[T]_N-T\vec b,[0]_N,\tau/T^2\right) = \frac{1}{(2 \pi \i)^2} \int_{\widetilde{\Sigma}_N} \d s \int_{\widetilde\ell_N} \d t \frac{\widetilde W_N^{\rm Mixed}(t;[1]_N-\vec b,[0]_N,\tau,T)}{ \widetilde W_N^{\rm Mixed}(s;[1]_N-\vec b,[0]_N,\tau,T)} \frac{{\rm e}^{-tx + sx'}}{t-s},\]
with
\[\widetilde W_N^{\rm Mixed}(z;[1]_N-\vec b,[0]_N,\tau,T)=\frac{\e^{\tau z^2/2}\prod_{k=1}^N\Gamma((1 -b_k- z)T)}{\Gamma(zT)^N}=\frac{z^N\e^{\tau z^2/2}}{\prod_{k=1}^N(1-b_k-z)} \frac{\prod_{k=1}^N\Gamma(1 + (1 -b_k- z) T)}{\Gamma(1+zT)^N},\] 
and where
\[ \tilde \ell_N := \frac{1}{T}\ell_N, \quad \widetilde \Sigma_N := \frac{1}{T}\Sigma_N .\]
It is easy to see that 
\[\lim_{T\to 0}\frac{\widetilde W_N^{\rm Mixed}(v;[T]_N-T\vec b,[0]_N,\tau,T)}{\widetilde W_N^{\rm Mixed}(u;[T]_N-T\vec b,[0]_N,\tau,T)}=\frac{W_N^{\rm GLUE+}(v;\vec b,\tau)}{W_N^{\rm GLUE+}(u;\vec b,\tau)}.\]
In analogy to the zero temperature results for the Log Gamma polymer and the O'Connell-Yor polymer, we obtain the following.
\begin{proposition}\label{prop:mixed}
{We have the point-wise limit
	\be\label{eq:zerotemplimitmixed}\lim_{T\to 0}\frac{1}{T}L_{N}^{{\rm Mixed}}\left(\frac{x}{T},\frac{x'}{T};[T]_N-T\vec b,[0]_N,\tau/T^2\right)=L_{N}^{\rm GLUE+}\left(x,x';\vec b,\tau\right),\qquad x,x'\in\mathbb R,\ee
	and there exists $C_{N}(\vec b,\tau)$ such that for $\epsilon>0$ sufficiently small and $T>0$ sufficiently small,
	\be\label{eq:dominatedcvgencemixed}\left|\frac{1}{T} L_{N}^{{\rm Mixed}}\left(\frac{x}{ T},\frac{x'}{T}; [T]_N-T\vec b,[0]_N, \tau/T^2\right)\right|\leq C_{N}(\vec b,\tau) h(x)\tilde h(x'), \ee
with $h,\tilde h$ as in \eqref{def:h2}.	}
\end{proposition}

\begin{corollary}\label{cor:mixed}	As $T\to\infty$, we have the convergence
	\[\mathbb E \left[ {\rm e}^{-\e^{t/T} Z_{N,N}^{{\rm Mixed}}\left( [T]_N-T\vec b,[0]_N,\tau/T^2\right)}\right]\longrightarrow \mathbb P_{\mathrm{GLUE+}}\left(\max\{x_1,\ldots, x_N\}\leq -t; \vec b, \tau \right),\]
	where the expectation on the left is with respect to the distribution of the mixed polymer partition function, and the probability on the right is with respect to the distribution of the eigenvalues $x_1,\ldots, x_N$ of the random matrix $Q$ constructed above as a sum of a GUE with  an LUE with external source matrix.
	\end{corollary}
}	
\begin{remark}To the best of our knowledge, the appearance of the sum of an LUE with external source matrix with a GUE matrix is completely new. 
{One can also think of the sum $Q=M+\sqrt{\tau N}H$ as the matrix obtained by applying independent Brownian motions to the independent entries of $M$. The resulting process for the eigenvalues of this process is called Dyson's Brownian motion, see e.g.\ \cite{Johansson}.}
In the zero temperature limit, $\tau$ thus plays the role of time in Dyson's Brownian motion applied to the eigenvalue configuration of a random matrix from the LUE with external source. If the measures $\d\mu_{N}^{\rm Mixed}(\vec x;\vec\alpha,\va,\tau)$ are positive polynomial ensembles, such an interpretation in terms of Dyson's Brownian motion remains valid for finite temperature: let $\tau'>\tau$, and suppose that $\vec x$ is a random configuration of points at time $\tau$; we then obtain a random configuration of points at time $\tau'$ by applying Dyson's Brownian motion to $\vec x$ for a time period $\tau'-\tau$. \end{remark}

\begin{remark}
There are other random matrix models whose eigenvalue distributions admit kernels of the form \eqref{def:Lhat}, but then in exponential variables, and with the special choice $a_m=m-1$. In this situation, the biorthogonal ensembles become polynomial ensembles, recall Remark \ref{remark:exp}.
This is the case for Wishart-type products of Ginibre matrices \cite{AkemannIpsenKieburg, AkemannKieburgWei, Forrester2, ForresterLiu, KuijlaarsStivigny, KuijlaarsZhang, ZyczkowskiSommers}, for Wishart-type products of truncated unitary matrices \cite{KieburgKuijlaarsStivigny}, and for Laguerre Muttalib-Borodin ensembles \cite{BorBiOE, ForresterWang, KuijlaarsZhang}. The different double contour integral representations for the kernels in these models can be found for instance in \cite[Appendix A]{ClaeysGirottiStivigny} or in \cite{ClaeysKuijlaarsWang}. We summarize the relevant choices of $W_N$ in these models in Table \ref{table}. \end{remark}

\begin{table}[h!]
\centering
\begin{tabular}{||c| c| c| c| c||} 
 \hline
 {\bf Model} &  ${\bf W_N(z)}$ &  ${\bf f_m(x)}$ & ${\bf g_m(x)}$ \\ [0.5ex] 
 \hline\hline
{\bf Log $\Gamma$ polymer} & $\frac{\prod_{j=1}^n\Gamma(\alpha_j-z)}{\prod_{k=1}^N\Gamma(z-a_k)}$ &  $\e^{a_m x}$& $\MeijerG*{n}{0}{0}{n + N}{-}{\valpha ; \vec 1 + \va-\vec e_m}{\e^{-x}}$  \\[1ex] 
 \hline
{\bf Homogeneous\, Log $\Gamma$ polymer} & $\frac{\Gamma(\alpha-z)^n}{\Gamma(z)^N}$ & $x^{m-1}$& $\MeijerG*{n}{0}{0}{n + N}{-}{\valpha ;(0)_{N-m+1};(1)_{m-1}}{\e^{- x}}$  \\[1ex] 
 \hline
{\bf OY polymer} & $\e^{\tau z^2/2}\frac{1}{\prod_{k=1}^N\Gamma(z-a_k)}$ & $\e^{a_m x}$& $\displaystyle\frac{1}{2 \pi \i} \int_{\frac{1}2 + \i \mathbb R} \frac{{\rm e}^{\tau \frac{t^2}2}{\e}^{-t x}}{(t - a_m) \prod_{k = 1}^N \Gamma(t - a_k)} \mathrm{d}t$  \\ [1ex]
 \hline
{\bf Homogeneous\, OY polymer} & $\e^{\tau z^2/2}\frac{1}{\Gamma(z-a)^N}$ &  $x^{m-1}$& $\displaystyle\frac{1}{2 \pi \i} \int_{\frac{1}2 + \i \mathbb R} \frac{{\rm e}^{\tau \frac{t^2}2}{\e}^{-(t - a)x}}{(t - a)^{N - m +1} \Gamma(t - a)^N} \mathrm{d}t$   \\ [1ex]
 \hline
{\bf Mixed polymer} & $\e^{\tau z^2/2}\frac{\prod_{j=1}^n\Gamma(\alpha_j-z)}{\prod_{k=1}^N\Gamma(z-a_k)}$ & $\e^{a_m x}$& $\displaystyle\frac{1}{2 \pi \i} \int_{\i \mathbb R} \frac{{\rm e}^{\tau \frac{t^2}2} \prod_{j = 1}^n \Gamma(\alpha_j - t)}{(t - a_m) \prod_{k = 1}^N \Gamma(t - a_k)} {\e}^{-tx} \mathrm{d}t $   \\ [1ex]\hline
{\bf Homogeneous\, Mixed polymer} & $\e^{\tau z^2/2}\frac{\Gamma(\alpha-z)^n}{\Gamma(z-a)^N}$ &  $x^{m-1}$ & $\displaystyle\frac{1}{2 \pi \i} \int_{\i \mathbb R} \frac{{\rm e}^{\tau \frac{t^2}2} \Gamma(\alpha - t)^n {\e}^{-(t - a)x}}{(t - a)^{N - m +1} \Gamma(t - a)^N} \mathrm{d}t$   \\ [1ex]
 \hline\hline
{\bf  GUE} & $z^N\e^{\tau z^2/2}$ & $x^{m-1}$& $x^{m-1}\e^{-\frac{x^2}{2\tau}}$  \\ [1ex]
 \hline
{\bf  GUE + external source} & $\prod_{m=1}^N(z-a_m)\ \e^{\tau z^2/2}$ &  ${x^{m-1}}$& ${{\e}^{-(\frac{x^2}{2\tau} - a_m x)}}$ \\ [1ex]
 \hline
{\bf LUE}& $\frac{z^N}{(z-1)^{N+\nu}}$ &  $x^{m-1}$& $x^{m-1+\nu}\e^{-x}$  \\[1ex]
 \hline
 {\bf LUE + external source}& $\frac{\prod_{m=1}^N(z-b_m)}{(z-1)^{N+\nu}}$ & $x^{m-1}$& {$x^\nu{\e}^{-(1 - b_m)x}$} \\[1ex]
 \hline
{\bf (LUE + external\, source) + GUE}& $ \frac{{\rm e}^{\tau z^2/2}z^N}{\prod_{k = 1}^N (z-1+b_k)}$ & $x^{m-1}$& $\displaystyle\frac{1}{2 \pi \i} \int_{\i \mathbb R} \frac{{\rm e}^{\tau \frac{t^2}2}t^{m-1} }{\prod_{j = 1}^N (t-1+b_j) } {\e}^{-tx} \mathrm{d}t $ \\[1ex]
 \hline
{\bf Ginibre products} & $\frac{\prod_{k=0}^n\Gamma(1+\nu_k-z)}{\Gamma(1-N-z)}$  & $\e^{(m-1)x}$& $\MeijerG*{n}{0}{0}{n}{-}{\vec\nu}{\e^{-x}}$  \\[1ex] 
 \hline
{\bf Muttalib-Borodin LUE} & $\frac{\Gamma(1-z)\Gamma(\nu+1-\theta z)}{\Gamma(1-N-z)}$ & $\e^{(m-1)x}$& $\e^{\theta(m-1)x}\e^{-\e^{x}}$  \\ [1ex]
 \hline
{\bf Truncated\, unitary products} & $\frac{\prod_{k=0}^n\Gamma(1+\nu_k-z)}{\prod_{k=0}^n\Gamma(1+\ell_k-N-z)}$  & $\e^{(m-1)x}$& $\MeijerG*{n}{0}{n}{n}{\vec\mu+\vec\nu}{\vec\nu}{\e^{-x}}$  \\ [1ex]
 \hline
 \end{tabular}
\caption{List of polymer and random matrix models associated to biorthogonal measures with symbol $F(z)$, and a possible choice, not necessarily biorthogonal, for the corresponding functions $f_m(x)$ and $g_m(x)$ in \eqref{def:biO}, for $m=1,\ldots, N$. The choices for $g_m$ in the GUE with external source and LUE with external source are only valid if $a_1,\ldots, a_N$ respectively $b_1,\ldots, b_N$ are distinct.
The last three models correspond to the case $a_m=m-1$, and the eigenvalues or squared singular values in these models correspond to the exponential variable  $y=\e^x$.}\label{table}
\end{table}

\subsection*{Acknowledgements} We thank Ivan Corwin and Milind Hegde for valuable conversations on Fredholm determinant representations of polymer partition functions, which constituted the starting point of this article. We are also grateful to Matteo Mucciconi and Jiyuan Zhang for useful discussions. M.C.
acknowledge support by the IRP ``PIICQ'',
funded by the CNRS, and the Centre Henri Lebesgue, program ANR-11-LABX-0020-0. T.C. acknowledges support by FNRS Research
Project T.0028.23, and by FNRS
under EOS project O013018F.

\end{document}